\documentclass[journal]{IEEEtran}
\usepackage{amssymb}
\usepackage{amsfonts}
\usepackage{savesym}
\usepackage{amsfonts,subfigure,multicol,color,verbatim,
graphicx,cite,epsfig,amssymb,amsmath,cases,bm,algorithm,
algorithmic, xcolor,multirow,float} 
\usepackage{amsmath}
\usepackage{array}
\usepackage{enumerate}
\usepackage{bm}
\savesymbol{iint}

\renewcommand{\IEEEQED}{\IEEEQEDopen}
\restoresymbol{TXF}{iint}

\newtheorem{theorem}{\textbf{Theorem}}
\newtheorem{definition}{\textbf{Definition}}
\newtheorem{lemma}{\textbf{Lemma}}

\newtheorem{remark}{\textbf{Remark}}

\IEEEoverridecommandlockouts

\begin{document}
\markboth{IEEE Trans. Veh. Tech.} {Li et al: Opt. HC-RAN\ldots}

\title{Energy-Efficient Joint Congestion Control and Resource Optimization in Heterogeneous Cloud Radio Access Networks}

\author{Jian~Li,~Mugen~Peng$^{\dagger}$,~\IEEEmembership{Senior
Member,~IEEE},~Yuling~Yu, and Zhiguo~Ding
\thanks{Jian~Li (e-mail: {\tt lijian.wspn@gmail.com}), Mugen~Peng (corresponding author, e-mail:
{\tt pmg@bupt.edu.cn}), ~Yuling~Yu (e-mail: {\tt
aliceyu1215@gmail.com}) are with the Key Laboratory of Universal
Wireless Communications for Ministry of Education, Beijing
University of Posts and Telecommunications, China.
Zhiguo~Ding (e-mail: {\tt
z.ding@lancaster.ac.uk}) is with the School of Computing and Communications, Lancaster University, LA1 4WA, UK.} }
\date{\today}

\maketitle

\begin{abstract}
The heterogeneous cloud radio access network (H-CRAN) is a promising
paradigm which integrates the advantages of cloud radio access
network (C-RAN) and heterogeneous network (HetNet). In this paper, we study the
joint congestion control and resource optimization to explore the
energy efficiency (EE)-guaranteed tradeoff between throughput
utility and delay performance in a downlink slotted H-CRAN. We
formulate the considered problem as a stochastic optimization
problem, which maximizes the utility of average throughput and
maintains the network stability subject to required EE constraint
and transmit power consumption constraints by traffic admission
control, user association, resource
block allocation and power allocation. Leveraging on the Lyapunov
optimization technique, the stochastic optimization problem can be
transformed and decomposed into three separate subproblems which can
be solved concurrently at each slot. The third mixed-integer
nonconvex subproblem is efficiently solved utilizing the continuity
relaxation of binary variables and the Lagrange dual decomposition
method. Theoretical analysis shows that the proposal can
quantitatively control the throughput-delay performance tradeoff
with required EE performance. Simulation results consolidate the
theoretical analysis and demonstrate the advantages of the proposal
from the prospective of queue stability and power consumption.
\end{abstract}


\section{Introduction}

Recently, the mobile operators are facing the continuously growing
demand for ubiquitous high-speed wireless access and the explosive
proliferation of smart phones. Justified by the urgent trend and
projecting the demand a decade ahead, a so-called 100 times spectral
efficiency (SE) boost and 1000 times energy efficiency (EE)
improvement compared to the current fourth generation (4G) wireless
systems are required\cite{report}. The increasingly demands make it
more challenging for operators to manage and operate wireless
networks and provide required quality of service (QoS) efficiently.
Therefore, the fifth generation (5G) wireless networks are expected
to fulfill these goals by putting forward new wireless network
architectures, advanced signal processing, and networking
technologies\cite{5G1,5G2}.

{By leveraging cloud computing technologies, the cloud radio access
network (C-RAN) has emerged as a promising solution for providing
good performance in terms of both SE and EE across software defined
wireless communication networks\cite{RRHcmcc}. In C-RANs, the remote
radio heads (RRHs) configured only with some front radio frequency
(RF) functionalities are connected to the baseband unit (BBU) pool
through fronthaul links (e.g., optical fibers) to enable cloud
computing-based large-scale cooperative signal processing. However,
the constrained fronthaul link between the RRH and the BBU pool
presents a performance bottleneck to large-scale cooperation gains.
Furthermore, real-time voice service and control signalling are not
efficiently supported in C-RANs. Therefore, the traditional C-RAN
must be enhanced and even evolved.

Motivated by solving the aforementioned challenges, the
heterogeneous cloud radio access network (H-CRAN) has been proposed
to combine the advantages of both C-RANs and heterogeneous networks
(HetNets) in our prior works \cite{HCRAN}\cite{HCRAN1}. As shown in
Fig. 1, the high power nodes (HPNs) are configured with the entire
communication functionalities from physical to network layers, and
the delivery of control and broadcast signalling is shifted from
RRHs to HPNs, which alleviates the capacity and time delay
constraints on the fronthaul. The BBU pool is interfaced to the HPN
for the inter-tier interference coordination. H-CRANs decouple the
control and user planes, and support the adaptive signaling/control
mechanism between connection-oriented and connectionless modes,
which can achieve significant overhead savings in the radio
connection/release. For a time-varying H-CRAN that adopts orthogonal
frequency division multiple access (OFDMA), besides of power and
resource block (RB) allocation, the traffic admission, the user
association are also critical for improving key performances. }

\begin{figure}
\centering  \vspace*{0pt}
\includegraphics[scale=0.55]{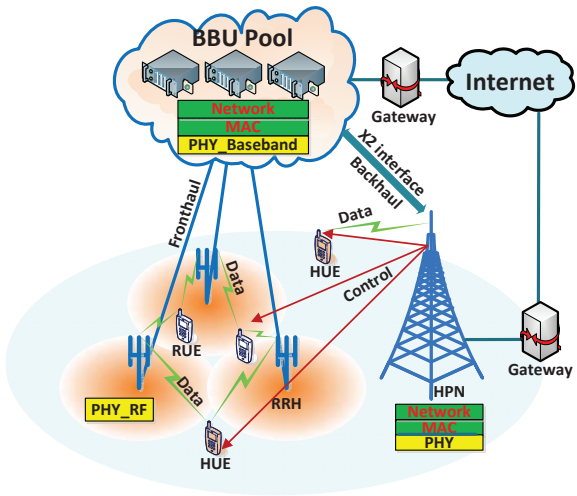}
\setlength{\belowcaptionskip}{-100pt} \vspace*{-2pt}\caption{The
heterogeneous cloud radio access networks (H-CRANs)}
\label{fig_HCRAN}
\end{figure}

\subsection{Related Works}
{

The resource optimization is significantly important to highlight
the great potentials of H-CRANs. The EE performance metric has
become a new design goal due to the sharp increase of the carbon
emission and operating cost of wireless communication
systems\cite{DFeng}. The tradeoff relationship between EE and SE has
attracted growing interests in the design of energy-efficient radio
resource optimization algorithms for wireless communication
systems\cite{XHong}\cite{OOnireti}. Power allocation was studied
in\cite{CHe} to address the EE-SE tradeoff in the downlink
multi-user distributed antenna system (DAS) with proportional rate
constraints. The authors of\cite{KCheung,WJing,CXiong} jointly
consider multi-dimensional resource optimization, such as
beamforming optimization, power allocation, and RB assignment, to
explore the EE-SE tradeoff in OFDMA networks. The EE-SE tradeoff
also got deep investigation in Device-to-Device (D2D) communications
and relay-aided cellular networks\cite{SHuang,IKu,ZZhou}.

However, the aforementioned literatures are typically based on the
full buffer assumptions and snapshot-based models. This indicates
that the stochastic and time-varying features of traffic arrivals
are not considered into the formulations. Therefore, only the
physical layer performance metrics such as SE and EE are optimized
and the resulting control policy is only adaptive to channel state
information (CSI). In practice, delay is also a key metric to
measure the QoS, which is also neglected in these literatures.

Contrary to the static models used in the EE-SE tradeoff, the
power-delay tradeoff is usually investigated from the long-term
average perspective in a time-varying system. The authors
of\cite{VKNLau} and\cite{YCui} aimed to dynamically optimize the
power and subband allocations to achieve the Pareto optimal tradeoff
between power consumption and average delay in OFDMA systems.
In\cite{HKChung}, a theoretical framework was presented to analyze
power-delay tradeoff in a time-varying OFDMA system with imperfect
CSIT. Adaptive antenna selection and power allocation were exploited
in\cite{YLiMSheng} to compromise the power consumption and average
delay in downlink distributed antenna systems. By devising
delay-aware beamforming algorithm, \cite{FZhang} studied the
power-delay tradeoff in multi-user MIMO systems.

As a common feature, \cite{VKNLau,YCui,HKChung,YLiMSheng,FZhang}
assumed that the random traffic arrival rate is inside the network
capacity region. This indicates that admission control is
unnecessary. Moreover, throughput was not formulated into the
problem, thus the results for power-delay tradeoff can hardly give
insights into energy-efficient resource optimization problems. }

\subsection{Main Contributions}

{In this paper, considering the traffic admission control, the
throughput, as in\cite{neelyinfocom} and\cite{HjuTWC}, is defined as
the maximum amount of admissible traffic that H-CRANs can stably
carry, and therefore it to some extent reflects SE. Based on this,
we try to incorporate throughput, delay, and EE into a theoretical
framework, and effectively balance throughput and delay when certain
EE requirement is guaranteed for any traffic arrival rate in slotted
H-CRANs. The major contributions of this paper are twofold.

\begin{itemize}
\item The congestion control is incorporated into the radio resource
optimization model for slotted H-CRANs without prior-knowledge of
random traffic arrival rates and channel statistics. The decomposed
subproblems can be solved concurrently at each slot with the online
observation of traffic queues and virtual queues, which have been
determined by the joint optimization results at the previous slot.
Simulations demonstrate the advantages of the proposal from the
prospective of queue stability and power savings.
\item Using the framework of Lyapunov optimization, we put forward a
formulation to quantitatively strike a balance between average
throughput and average delay, meanwhile guarantee the required EE
performance of H-CRANs. Only by adjusting a control parameter, the
proposal provides a controllable method to balance the
throughput-delay performance on demand, which in turn adaptively
affects admission control and resource allocation.
\end{itemize}
}

The rest of this paper is organized as follows. Section II will
describe the system model and formulate the stochastic optimization
problem. Based on the framework of Lyapunov optimization, the
stochastic optimization problems of traffic admission control,
user association, RB and power allocation will be transformed and
decomposed in Section III. The challenging subproblem for user
association, RB and power allocation will be solved in Section IV.
The performance bounds of proposal will be analyzed in section V.
Numerical simulations will be shown in Section VI. Finally, Section
VII will summarize this paper.

\section{System Model and Problem Formulation}

In this section, we begin with describing the physical layer model,
followed by introducing the queue dynamics and the queue stability.
We then formally formulate the stochastic optimization problem. {For
convenience, the notations used are listed in Table I.}

\begin{table}[ht]
\renewcommand\arraystretch{0.6}
\centering \caption{{Summary of Notations}} {
\begin{tabular}{|m{1.0cm}|m{6.5cm}|}
\hline
\textbf{Notation} & \textbf{Description}  \\
\hline $\mathcal {R}$ & Set of RRHs \\
\hline $\mathcal {U} _H$ & Set of HUEs \\
\hline $\mathcal {U} _R$ & Set of RUEs \\
\hline $\mathcal {K} _H$ & Set of RBs used by RRHs \\
\hline $\mathcal {K} _R$ & Set of RBs used by the HPN \\
\hline $s_m(t)$ & Association indicator for the HUE $m$ at the time slot $t$\\
\hline $g_{ijk}(t)$ & CSI on the RB $k$ from the RRH $i$ to the RUE $j$ at the slot $t$\\
\hline $g_{imk}(t)$ & CSI on the RB $k$ from the RRH $i$ to the HUE $m$ at the slot $t$ \\
\hline $g_{ml}(t)$ & CSI on the RB $l$ from the HPN to the HUE $m$ at the slot $t$ \\
\hline ${c_l^H}(t)$ & RB usage indicator for the RB $l$ of the HPN at the slot $t$\\
\hline ${c_k^R}(t)$ & RB usage indicator for the RB $k$ of RRHs at the slot $t$\\
\hline $p_{ijk}(t)$ & Allocated power for the RUE $j$ occupying the RB $k$ of the RRH $i$ at the slot $t$ \\
\hline $p_{imk}(t)$ & Allocated power for the HUE $m$ occupying the RB $k$ of the RRH $i$ at the slot $t$ \\
\hline $p_{ml}(t)$ &  Allocated power for the HUE $m$ occupying the RB $l$ of the HPN at the slot $t$ \\
\hline ${a_{jk}}(t)$ & Allocation of the RB $k$ of RRHs to the RUE $j$ at the slot $t$\\
\hline ${a_{mk}}(t)$ & Allocation of the RB $k$ of RRHs to the HUE $m$ at the slot $t$ \\
\hline ${b_{ml}}(t)$ & Allocation of the RB $l$ of the HPN to the HUE $m$ at the slot $t$ \\
\hline ${\mu _m}(t)$ & Transmit rate of the HUE $m$ at the slot $t$ \\
\hline ${\mu _j}(t)$ & Transmit rate of the RUE $j$ at the slot $t$ \\
\hline $R_m(t)$ & The amount of admitted traffics for the HUE $m$ at the slot $t$  \\
\hline $R_j(t)$ & The amount of admitted traffics for the RUE $j$ at the slot $t$  \\
\hline $Q_m(t)$ & Traffic buffering queue length for the HUE $m$ at the slot $t$ \\
\hline $Q_j(t)$ & Traffic buffering queue length for the RUE $j$ at the slot $t$ \\
\hline $\gamma _{m}(t)$ & Auxiliary variable for the throughput of the HUE $m$ at the slot $t$ \\
\hline $\gamma _{j}(t)$ & Auxiliary variable for the throughput of the RUE $j$ at the slot $t$ \\
\hline $H_m(t)$ & Virtual queue length for the HUE $m$ with arrival $\gamma _{m}$ at the slot $t$ \\
\hline $H_j(t)$ & Virtual queue length for the RUE $j$ with arrival $\gamma _{j}$ at the slot $t$  \\
\hline ${Z}(t)$ & Virtual queue length for the average EE constraint at the slot $t$\\
\hline $x_{mk}$ & Continuous auxiliary variable for RB allocation ${a_{mk}}$\\
\hline $y_{ml}$ & Continuous auxiliary variable for RB allocation ${b_{ml}}$\\
\hline $w_{ijk}$ & Auxiliary variable for power allocation $p_{ijk}$ \\
\hline $v_{imk}$ & Auxiliary variable for power allocation $p_{imk}$ \\
\hline $u_{ml}$ & Auxiliary variable for power allocation $p_{ml}$ \\
\hline
\end{tabular}}
\end{table}

\subsection{Physical Layer Model}

The downlink transmission in an OFDMA-based H-CRAN is considered, in
which one HPN and $N$ RRHs are consisted. Since the HPN is mainly
used to deliver the control signalling and guarantee the basic
coverage for certain area, the UEs with low traffic arrival rates
are more likely served by HPN and they are labeled as HUEs.
Meanwhile, since the RRHs are efficient to provide high bit rates,
the user equipments (UEs) with high traffic arrival rates will be served by the RRHs.
Let $\mathcal {R} = \{1, 2, ..., N\}$ denote the set of RRHs, let
$\mathcal {U} _H$ denote the set of HUEs and let $\mathcal {U} _R$
denote the set of RUEs. To completely avoid the severe inter-tier
interferences, the RBs of H-CRAN are partitioned and assigned
respectively to the RRH and the HPN tiers. Let $\mathcal {K}_R$ and
$\mathcal {K}_H$ denote the set of RBs used by RRH tier and HPN
tier, respectively. Let $W$ and $W_0$ denote the system bandwidth
and the bandwidth of each RB, respectively. Any UE that is
associated with RRH tier receives signal simultaneously from
multiple cooperative RRHs on allocated RBs, and the RBs allocated to
different UEs are orthogonal, it is thus inter-RRH interference-free
among UEs. The network is assumed to operate in slotted time with
slot duration $\tau$ and indexed by $t$.

The HUEs can be associated with RRH tier to get more transmission
opportunity when the traffic load of HPN becomes heavier, while the
RUEs are served only by RRHs, which is usually in accordance with
the practice. The user association strategy plays an important role
in improving the utilization efficiencies of limited radio resources
in an H-CRAN. Let the binary variable $s_m(t)$ indicate the user
association of the HUE $m$ at the slot $t$, which is 1 when the HUE
$m$ is associated with the RRH tier or 0 when it is associated with
the HPN. Let $g_{ijk}(t)$, $g_{imk}(t)$ and $g_{ml}(t)$ represent
the CSIs on the RB $k$ from the RRH $i$ to the RUE $j$, the RB $k$
from the RRH $i$ to the HUE $m$, the RB $l$ from the HPN to the HUE
$m$, respectively. {Note that these CSIs account for the antenna
gain, beamforming gain, path loss, shadow fading, fast fading, and
noise together. In this paper, as the RB allocation for both RRHs
and HPN tiers in OFDMA-based H-CRANs are focused, the antenna
configuration and beamforming design are not specified, and neither
of them affects the general formulation.} The CSIs are assumed to be
independently and identically distributed (i.i.d.) over slots, and
takes values in a finite state place. Let $p_{ijk}(t)$ denote the
allocated transmit power for RUE $j$ on RB $k$ from RRH $i$ at slot
$t$, let $p_{imk}(t)$ denote the allocated transmit power for HUE
$m$ on RB $k$ from RRH $i$ if HUE $m$ is associated with RRH tier at
slot $t$, and let $p_{ml}(t)$ denote the allocated transmit power
for HUE $m$ on RB $l$ from HPN if HUE $m$ is associated with HPN at
slot $t$. Furthermore, let the binary variable ${a_{jk}}(t)$ and
${a_{mk}}(t)$ indicate the allocation of RB $k$ of RRH tier to RUE
$j$ and HUE $m$ at slot $t$, respectively, and let ${b_{ml}}(t)$
indicate the allocation of RB $l$ of HPN to HUE $m$ at slot $t$,
then we have the following non-reuse constraints
\begin{equation}
{c_k^R}(t) = \sum\limits_{j \in {\mathcal {U}_R}} {{a_{jk}}(t)} +
\sum\limits_{m \in {\mathcal {U}_{H}}} {{s_m(t)}{a_{mk}}(t)} \le
1,\label{reuse}
\end{equation}
\begin{equation}
c_l^H(t) = \sum\limits_{m \in \mathcal {U} _H}
{(1-s_m(t)){b_{ml}}(t)} \le 1.
\end{equation}
for the RRH tier and HPN tier, respectively.

For the UEs that are served by RRHs (including all the RUEs and some
associated HUEs), the maximum ratio combining (MRC) is assumed to be
adopted. Therefore, the transmit rate of RUE $j$ and HUE $m$ at slot
$t$ is given by
\begin{equation}
{\mu _j}(t) = \sum\limits_{k \in \mathcal {K}_R}
{{a_{jk}(t)}{W_0}{{\log }_2}(1 + \sum\limits_{i \in \mathcal {R}}
{{p_{ijk}(t)}{g_{ijk}(t)}} )},
\end{equation}
\begin{equation}
\begin{array}{l}
{\mu _m}(t) = (1-{s_m}(t))\!\!\sum\limits_{l \in \mathcal {K}_H}
\!\!\!{{b_{ml}}(t){W_0}{{\log }_2}(1 + {g_{ml}}(t){p_{ml}}(t))} \\
~~~~~~~+ {s_m}(t)\!\!\sum\limits_{k \in \mathcal {K}_R}\!\!\!
{{a_{mk}}(t){W_0}{{\log }_2}(1 + \sum\limits_{i \in \mathcal {R}}
{{p_{imk}}(t){g_{imk}}(t)} )}, \end{array}
\end{equation}
respectively. Accordingly, the total transmit rate of the network is
given by
\begin{equation}
{\mu _{\rm sum}}(t) = \sum\limits_{m \in {\mathcal {U}_H}} {{\mu
_m}(t)} + \sum\limits_{j \in {\mathcal {U}_R}} {{\mu _j}(t)},
\end{equation}

With the resource allocations, the transmit power of RRH $i$ and HPN
is given by
\begin{equation}
{p_i}(t)\! =\!\! \sum\limits_{j \in {\mathcal {U}_R}}\!
{\sum\limits_{k \in \mathcal {K}_R}\!\! {{a_{jk}}(t){p_{ijk}}(t)} }
+\!\! \sum\limits_{m \in {\mathcal {U}_H}} \!{\sum\limits_{k \in
\mathcal {K}_R}\!\! { {s_m}(t){a_{mk}}(t){p_{imk}}(t)} },
\end{equation}
\begin{equation}
{p_H}(t) = \sum\limits_{m \in {\mathcal {U}_H}} {\sum\limits_{l \in
\mathcal {K}_H} {(1-s_m(t))}{{b_{ml}}(t){p_{ml}}(t)} },
\end{equation}
respectively. Accordingly, the total power consumption of the
network is given by
\begin{equation}
{p_{\rm sum}}(t) = \sum\limits_{i \in \mathcal {R}} {\varphi _{\rm
eff}^R{p_i}(t)} + p_{\rm c}^{R} + \varphi _{\rm eff}^Hp_H(t) +
p_{\rm c}^H,
\end{equation}
where $\varphi _{\rm eff}^R$ and $\varphi _{\rm eff}^H$ are the
drain efficiency of RRH and HPN, respectively, $p_{\rm c}^{\rm{R}}$
and $p_{\rm c}^{\rm{H}}$ are the static power consumption of RRHs
and HPN, respectively, including the circuit power, the fronthaul
power consumption and the backhaul power consumption.

\subsection{Queue Dynamics and Queue Stability}

In the considered H-CRAN, separate buffering queues are maintained
for each UE. Let ${Q_{m}}(t)$ and ${Q_{j}}(t)$ denote the length of
buffering queues maintained for HUE $m$ and RUE $j$, respectively.
Let ${A_{m}}(t)$ and ${A_{j}}(t)$ denote the amount of random
traffic arrivals at slot $t$ destined for HUE $m \in \mathcal {U}_H$
and RUE $j \in \mathcal {U}_R$, respectively.

{\emph{\textbf{Assumption} 1 (Random Traffic Arrivals Model)}:
Assume that $A_m(t)$ and $A_{j}(t)$ are i.i.d. over time slots
according to a general distribution, which are independent w.r.t.
$m$ and $j$. Furthermore, there exists certain peak amount of
traffic arrivals $A_{m}^{\max}$ and $A_{j}^{\max}$, respectively,
which satisfying $A_m(t) \le A_m^{\max}$ and $A_{j}(t) \le
A_{j}^{\max}$.}

In practice, the statistics of $A_m(t)$ and $A_{j}(t)$ are usually
unknown to H-CRANs, and the achievable capacity region is usually
difficult to estimate, the situation that the exogenous arrival
rates are outside of the network capacity region may occur. In this
situation, the traffic queues cannot be stabilized without a
transport layer flow control mechanism to limit the amount of data
that is admitted. To this end, the H-CRAN tries to maximize its
utility by admitting as many traffic datas as possible, and to
minimize the penalty from traffic congestion by transmitting as many
traffic datas as possible with the limited radio resources. Let
$R_m(t)$ and $R_{j}(t)$ denote the amount of admitted traffic datas
out of the potentially substantial traffic arrivals for HUE $m$ and
RUE $j$, respectively. Therefore, the traffic buffering queues for
HUE $m$ and RUE $j$ evolve as
\begin{equation}
{Q_{m}}(t + 1) = \max\{{Q_{m}}(t) -{\mu _{m}}(t)\tau, 0 \}+
{R_{m}}(t),\label{dynamicsQ0m}
\end{equation}
\begin{equation}
{Q_{j}}(t + 1) = \max\{{Q_{j}}(t) - {\mu _{j}}(t)\tau, 0\} +
{R_{j}}(t),\label{dynamicsQij}
\end{equation}
respectively, where we have $0 \le R_m(t) \le A_m(t)$ and $0 \le
R_{j}(t) \le A_{j}(t)$ at each slot.

To model the impacts of joint congestion control and resource
allocation on the average delay and the achieved throughput utility,
the definition of network stability will be formally given.

\begin{definition}
\emph{A single discrete time queue $Q(t)$ is strongly stable if
\begin{equation}
\mathop {\lim \sup }\limits_{T \to \infty }
\frac{1}{T}\sum\limits_{t = 0}^{T - 1} {\mathbb{E}[Q(t)]} < \infty
\end{equation}}
\end{definition}
\begin{definition}
\emph{A network of queues is strongly stable if all the individual
queues of the network are strongly stable.}
\end{definition}

To guarantee certain EE requirement when dynamically
making congestion control and resource optimization for the H-CRAN.
As in{\cite{yzli0}}, the definition of EE is given as follows.

\begin{definition}
\emph{The EE of considered H-CRAN is defined as the ratio of the
long-term time averaged total transmit rate to the corresponding
long-term time averaged total power consumption in the unit of
bits/Hz/J, which is given by
\begin{equation}
{\eta _{EE}} = \frac{{\mathop {\lim }\limits_{T \to \infty }
\frac{1}{T}\sum\limits_{t = 0}^{T - 1} {\mathbb{E}[{\mu _{sum}}(t)]}
}}{W{\mathop {\lim }\limits_{T \to \infty }
\frac{1}{T}\sum\limits_{t = 0}^{T - 1} {\mathbb{E}[{p_{sum}}(t)]} }}
= \frac{{{{\bar \mu }_{sum}}}}{{{W}{{\bar p}_{sum}}}},
\end{equation}}
\end{definition}

A queue is strongly stable if it has a bounded time average queue
backlog. According to the Little's Theorem{\cite{yingcui}}, the
average delay is proportional to the average queue length for a
given traffic arrival rate. Furthermore, when a network of traffic
queues is strongly stable, the average achieved throughput can be
given by the time averaged amount of admitted exogenous traffic
arrivals. Therefore, the average throughput for HUE and RUE is
expressed as
\begin{equation}
{{\bar r}_m} = \mathop {\lim }\limits_{T \to \infty }
\frac{1}{T}\sum\limits_{t = 0}^{T - 1} {{R_m}(t)},
\end{equation}
\begin{equation}
{{\bar r}_{j}} = \mathop {\lim }\limits_{T \to \infty }
\frac{1}{T}\sum\limits_{t = 0}^{T - 1} {{R_{j}}(t)},
\end{equation}
respectively.

\subsection{Problem Formulation}

The profit brought by dynamic joint congestion control and resource
optimization can be characterized by the utility of average
throughput, which is given by
\begin{equation}
U({\bar {\bf{r}}}) = \alpha {\sum\limits_{j \in {\mathcal {U}_R}}
{{g_R}({{\bar r_{j}}} )} } + \beta \sum\limits_{m \in \mathcal
{U}_H} {{g_H}({{\bar r_{m}}} )},
\end{equation}
where $\bar {\bf{r}} =[\bar r_m, \bar r_{j} : m \in \mathcal {U}_H,
j \in \mathcal {U}_R]$ is the vector of average throughput for all
UEs, ${g_R}(.)$ and ${g_H(.)}$ are the non-decreasing concave
utility function for RUEs and HUEs, respectively, $\alpha $ and
$\beta $ are the positive utility prices which indicate the relative
importance of corresponding utility functions.

Let ${\bf{r}} = [ {R_{j}(t)},{R_{m}(t)} : j \in \mathcal {U}_R, m
\in \mathcal {U}_H]$, ${\bf{s}}(t) = [s_m(t):m \in \mathcal {U}_H]$,
${\bf{p}} = [ {p_{ijk}(t)}, p_{imk}(t), p_{ml}(t): j \in \mathcal
{U}_R, m \in \mathcal {U}_H, i \in \mathcal {R}, k \in \mathcal
{K}_R, l \in \mathcal {K}_H]$, ${{\bf{a}}} = [{a_{jk}(t)},
{a_{mk}(t)}, b_{ml}(t) : j \in \mathcal {U}_R, m \in \mathcal {U}_H,
k \in \mathcal {K}_R, l \in \mathcal {K}_H]$ denote the vectors of
traffic admission, user association, power allocation and RB
allocation, respectively. To maximize the throughput utility of
networks and ensure the strong stability of traffic queues at the
same time by joint congestion control and resource optimization, the
stochastic optimization problem can be formulated as follows:
\begin{equation}
\begin{array}{l}
\mathop {\max }\limits_{\{ {\bf{r}},{\bf{s}},{\bf{p}},{{\bf{a}}}\} } U(\bar {\bf{r}}) \\
{\rm{s.t.}}~{\rm{C1}}: c_{k}^R(t) \le 1,\forall k,t,\\
~~~~~{\rm{C2}}: c_{l}^H(t) \le 1, \forall l,t,\\
~~~~~{\rm{C3}}: p_i(t) \le p_i^{\max},\forall i,t,\\
~~~~~{\rm{C4}}: p_H(t) \le p_H^{\max},\forall t,\\
~~~~~{\rm{C5}}:{\eta _{EE}} \ge \eta _{_{EE}}^{\rm req},\\
~~~~~{\rm{C6}}:Q_m(t)~{\rm{and}}~Q_{j}(t)~{\rm{are~strongly~stable}},\forall m,j,\\
~~~~~{\rm{C7}}: {R_{m}}(t) \le A_{m}(t), {R_{j}}(t) \le A_{j}(t),\forall m,j,t,\\
~~~~~{\rm{C8}}: {a_{jk}}(t), {a_{mk}}(t), {b_{ml}}(t), s_m(t) \in \{
0,1\},\forall j,k,m,l,t.
\end{array}\label{StochPro}
\end{equation}
where $p_i^{\max}$ and $p_H^{\max}$ denote the maximum transmit
power consumption of RRH $i$ and HPN, respectively, and $\eta_{\rm
EE}^{\rm req}$ denote the required EE of the network. C1 and C2
ensure that each RB of both tiers cannot be allocated to more than
one UE. C3 and C4 restrict the instantaneous transmit power of each
RRH and HPN. C5 makes the EE performance above predefined level. C6
ensures the queue stability to guarantee a finite average delay for
each queue. C7 ensures that the amount of admitted traffics cannot
be more than that of arrivals, C8 is the binary constraint for the
RB allocation and the user association.

For realistic H-CRANs, on one hand, the bursty traffic arrivals are
time-varying and unpredictable, and the key parameters are hardly
captured, which makes it infeasible to obtain optimal solution in an
offline manner; on the other hand, the dense deployment of RRHs in
H-CRANs exacerbates the computational complexity of centralized
solution. Therefore, an online and low-complexity solution to make
decisions effectively on user association, RB and power
allocation will be designed in the following sections.

\section{Dynamic Optimization Utilizing Lyapunov Optimization}

In response to the challenges of problem (\ref{StochPro}), we take
advantage of Lyapunov optimization techniques{\cite{neelyinfocom}}
to design an online control framework, which is able to make all
three important control decisions concurrently, including traffic
admission control, user association, RB and power allocation.

\subsection{Equivalent Formulation via Virtual Queues}

The formulated dynamic resource optimization problem in
(\ref{StochPro}) involves maximizing a non-decreasing concave
function of average throughputs, which is a bottleneck for solution.
To address this issue, the non-negative auxiliary variables
$\gamma_{m}(t)$ and $\gamma_{j}(t)$ are introduced to transform
problem (\ref{StochPro}) into an equivalent optimization problem
with a time averaged utility function of instantaneous throughputs
instead of a utility function of average throughputs. Let
$\bm{\gamma} = [\gamma_{m}(t),\gamma_{j}(t) : m \in \mathcal {U}_H,
j \in \mathcal {U}_R]$ be the vector of introduced auxiliary
variables, then we have the following equivalent problem:
\begin{equation}
\begin{array}{l}
\max\limits_{\{\bf{r},\bf{s},\bf{p},\bf{a}, \bm{\gamma}\}} \mathop
{\lim }\limits_{T \to \infty } \frac{1}{T}\sum\limits_{t = 0}^{T -
1} {U({\bm{\gamma}} (t))}\\
~~~~{\rm{s.t.}}~~~{\rm{C1-C8}},\\
~~~~~~~~~~~{\rm{C9}}:\gamma_{j}(t) \le A_{j}^{\max}, \gamma_{m}(t) \le A_{m}^{\max},\forall j,m,t,\\
~~~~~~~~~~~{\rm{C10}}:{{\bar \gamma }_{j}} \le {{\bar r}_{j}},
{{\bar \gamma }_{m}} \le {{\bar r}_{m}},\forall j,m.
\end{array}\label{EqvStochPro}
\end{equation}
where $U(\bm{\gamma}(t)) = \alpha \!{\sum\limits_{j \in {\mathcal
{U}_R}} \!{{g_R}({{\gamma _{j}}}(t))} } + \beta \!\!\!
\sum\limits_{m \in \mathcal {U}_H} \!\!{{g_H}({{\gamma _{m}}}
(t))}$, ${\bar \gamma _{j}} \!=\! \mathop {\lim }\limits_{T \to
\infty } \!\frac{1}{T}\!\sum\limits_{t = 0}^{T-1} \!{{\gamma
_{j}}(t)}$ and ${\bar \gamma _{m}}\!\! =\!\! \mathop {\lim
}\limits_{T \to \infty }\!\frac{1}{T}\!\sum\limits_{t = 0}^{T-1}
\!{{\gamma _{m}}(t)}$. Let ${\bf{r}}^{\rm opt}$, ${\bf{s}}^{\rm
opt}$, ${\bf{p}}^{\rm opt}$ and ${\bf{a}}^{\rm opt}$ denote the
optimal solution to the original problem (\ref{StochPro}), and let
$\bf{r}^*$, ${\bf{s}}^*$, ${\bf{p}}^*$, ${\bf{a}}^*$, and
$\bm{\gamma}^*$ denote the optimal solution to the equivalent
problem (\ref{EqvStochPro}), then we have the following theorem.
\begin{theorem}
\emph{The optimal solution for the transformed problem
(\ref{EqvStochPro}) can be directly turned into an optimal solution
for the original problem (\ref{StochPro}). Specifically, the optimal
solution for the original problem can be obtained as ${\bf{r}}^{\rm
opt} = {\bf{r}}^{*}$, ${\bf{s}}^{\rm opt} = {\bf{s}}^{*}$,
${\bf{p}}^{\rm opt} = {\bf{p}}^{*}$, and ${\bf{a}}^{\rm opt} =
{\bf{a}}^{*}$. }
\end{theorem}
\begin{proof}
Please refer to Appendix A.
\end{proof}

To ensure the average constraints for auxiliary variables in C10,
the virtual queues ${H_{m}}(t)$ and ${H_{j}}(t)$ are introduced for
each HUE and each RUE, respectively, and they evolve as
\begin{equation}
{H_{m}}(t + 1) = \max \{ {H_{m}}(t) - {R_{m}}(t),0\} + {\gamma
_{m}}(t),\label{vq0m}
\end{equation}
\begin{equation}
{H_{j}}(t + 1) = \max \{ {H_{j}}(t) - {R_{j}}(t),0\} + {\gamma
_{j}}(t),\label{vqij}
\end{equation}
where $H_{m}(0) = 0$, $H_{j}(0) = 0$, $ {\gamma _{m}}(t)$ and
${\gamma _{j}}(t)$ will be optimized at each slot.

Similarly, to ensure the EE performance constraint C5, the virtual
queue ${Z}(t)$ with initial value ${Z}(0) = 0$ is also introduced,
and it evolves as
\begin{equation}
Z(t + 1) = \max \{ Z(t) - {\mu _{\rm sum}}(t),0\} + {W}\eta _{{\rm
EE}}^{\rm req}{p_{\rm sum}}(t),\label{vqpower}
\end{equation}

Intuitively, the auxiliary variables ${\gamma _{m}}(t)$, ${\gamma
_{j}}(t)$ and ${W}\eta _{{\rm EE}}^{\rm req}{p_{\rm sum}}(t)$ can be
looked as the ¡°arrivals¡± of virtual queues ${H_{m}}(t)$,
${H_{j}}(t)$ and $Z(t)$, respectively, while ${R_{m}}(t)$,
${R_{j}}(t)$ and ${\mu _{\rm sum}}(t)$ can be looked as the service
rate of such virtual queues.
\begin{theorem}
\emph{The constraints C5 and C10 can be satisfied only when the
virtual queues $ {H_{m}}(t)$, ${H_{j}}(t)$ and ${Z}(t)$ are stable.}
\end{theorem}
\begin{proof}
Please refer to Appendix B.
\end{proof}

\subsection{Problem Transformulation via Lyapunov Optimization}

Let $\chi (t) = [
{Q_{m}}(t),{Q_{j}}(t),{H_{m}}(t),{H_{j}}(t),{Z}(t): m \in \mathcal
{U}_H, j \in \mathcal {U}_R]$ denote the vector of the traffic
queues and virtual queues. To represent a scalar metric of queue
congestion, the quadratic Lyapunov function is defined as
\begin{equation}
\begin{array}{l}
L(\chi (t))\! =\! \frac{1}{2}(\!\!\!\!\sum\limits_{m \in {\mathcal
{U}_H}}\!\!\!\!\!{Q_{m}^2(t)}\! + \!\!\!\!\!\sum\limits_{m \in
{\mathcal {U}_H}} {\!\!\!\!\!H_{m}^2(t)} \!+\!\!\!\! {\sum\limits_{j
\in {\mathcal {U}_R}} \!\!\!\!{Q_{j}^2(t)} }\! +\!\!\!\!
{\sum\limits_{j \in {\mathcal {U}_R}} {\!\!\!\!H_{j}^2(t) + } }
{\!Z^2(t)} ),
\end{array}
\end{equation}
where a small value of $L(\chi (t))$ implies that both actual queues
and virtual queues are small and the queues have strong stability.
Therefore, the queue stability can be ensured by persistently
pushing the Lyapunov function towards a lower congestion state. To
stabilize the traffic queues, while additionally satisfy some
average constraints and optimize the system throughput utility, the
Lyapunov conditional drift-minus-utility function is defined as
\begin{equation}
\begin{array}{l}
\Delta (\chi (t)) \!=\!
\mathbb{E}[ L(\chi (t + 1))\! - \!L(\chi (t))\! - \!VU({\bm{\gamma}}(t)) |\chi (t)], \\
\end{array}\label{drift}
\end{equation}
where the control parameter $V (V \ge 0)$ represents the emphasis on
utility maximization compared to queue stability. By adjusting $V$,
flexible design choices among various tradeoff points between queue
delay and throughput utility can be made by operators. With the
dynamics of practical traffic queues and introduced virtual queues,
the upper bound of drift-plus-utility is derived in the following
lemma.
\begin{lemma}
\emph{At slot $t$, for any observed queue state, the Lyapunov
drift-minus-utility of an H-CRAN with any joint congestion control
and resource optimization strategy satisfies the following
inequality,
\begin{equation}
\begin{array}{l}
\Delta (\chi (t)) \le C -\mathbb{E}\left[ { {\sum\limits_{j \in
{\mathcal {U}_R}} {( V\alpha {g_R}({\gamma _{j}(t)})-
{H_{j}}(t){\gamma _{j}}(t)) } } } \right.\\
~~~~~~~~~~~+\!\!\!\left. {\sum\limits_{m \in {\mathcal {U}_H}}
{\!\!\!( V\beta
{g_H}({\gamma_{m}(t)}) - {H_{m}}(t){\gamma _{m}}(t)} ) |\chi(t)} \right]\\
~~~~~~~~~~~- \mathbb{E}\left[ { \sum\limits_{m \in {\mathcal {U}_H}}
({{H_{m}}(t) - {Q_{m}}(t)) {R_{m}}(t) } } \right.\\
~~~~~~~~~~~\left. {+\!{\sum\limits_{j \in {\mathcal {U}_R}} ({
{H_{j}}(t) - {Q_{j}}(t)) {R_{j}}(t)|\chi (t) } } } \right]\\
~~~~~~~~~~~- \mathbb{E}\left[ { \sum\limits_{m \in {\mathcal {U}_H}}
\!\!{{Q_{m}}(t)\mu_{m}(t)\tau } + \!\!\!{\sum\limits_{j \in
{\mathcal{U}_R}}\!\!{{Q_{j}}(t){\mu _{j}}(t)\tau} } } \right.\\
~~~~~~~~~~~\left. {+ {{Z}(t)({\mu _{sum}}(t) - W\eta _{{\rm EE}}^{\rm req}{p_{sum}}(t)) } |\chi (t)} \right],\\
\end{array}\label{driftBound}
\end{equation}
where $C$ is a finite constant parameter that satisfies
\begin{equation}
\begin{array}{l}
C \ge \frac{1}{2}\mathbb{E}\left[ {{{(W\eta _{\rm EE}^{\rm req}{p_{sum}}(t))}^2} + \mu _{sum}^2(t) + \!\!\!\!\sum\limits_{j \in {\mathcal {U}_R}}\!\!\! {(2R_j^2(t) + } } \right. \\
\left. {\mu _j^2(t){\tau ^2} + \gamma _j^2) +\!\!\!\! \sum\limits_{m
\in {\mathcal {U}_H}}\!\!\! {(2R_m^2(t) + \mu _m^2(t){\tau ^2} +
\gamma _m^2)} |\chi (t)} \right].
\end{array}
\end{equation}
}
\end{lemma}
\begin{proof}
Please refer to Appendix C.
\end{proof}

According to the theory of Lyapunov optimization, instead of
minimizing the drift-minus-utility expression (\ref{drift})
directly, a good joint congestion control and resource optimization
strategy can be obtained by minimizing the right hand side (R.H.S.)
of (\ref{driftBound}) at each slot, which can be decoupled to a
series of independent subproblems and can be solved concurrently
with the real-time online observation of traffic queues and virtual
queues at each slot.

\subsection{Problem Decomposition}

\subsubsection{Auxiliary Variable Selection}
The optimal auxiliary variables can be obtained by minimizing the
first item of R.H.S. of (\ref{driftBound}) at each slot, i.e. $ -
{\sum\limits_{j \in {\mathcal {U}_R}} ({ V\alpha {g_R}({\gamma
_{j}(t)}) - {H_{j}}(t){\gamma _{j}}(t)) } } - \sum\limits_{m \in
{\mathcal {U}_H}} {( V\beta {g_H}({\gamma _{m}(t)}) -
{H_{m}}(t){\gamma _{m}}(t)} )$. Since the auxiliary variables are
independent among different UEs, the minimization can be decoupled
to be computed for each UE separately as
\begin{equation}
\begin{array}{l}
\mathop {\max }\limits_{{\gamma _{j}(t)}}~~ V\alpha {g_R}({\gamma
_{j}(t)}) - {H_{j}}(t){\gamma _{j}}(t) ~~~~\\
~{\rm{s.t.}}~~~ \gamma_{j}(t) \le A_{j}^{\max},~~~~~~~~~~~~~~~~~~~~
\end{array}\label{optimalauxij}
\end{equation}
\begin{equation}
\begin{array}{l}
\mathop {\max }\limits_{{\gamma _{m}(t)}} ~~V\beta {g_H}({\gamma
_{m}(t)})
- {H_{m}}(t){\gamma _{m}}(t)\\
~{\rm{s.t.}}~~~\gamma_{m}(t) \le A_{m}^{\max}.~~~~~~~~~~~~~~~~~~
\end{array}\label{optimalauxm}
\end{equation}

Apparently, the problems above are both convex optimization
problems. Therefore, the optimal auxiliary variables can be derived
by differentiating the objective function and make the result equal
to zero. In the case of logarithmic utility function, we have
${\gamma _{j}(t)} = \min \left[ { \frac{{V\alpha
}}{{{H_{j}(t)}}},A_{j}^{\max}} \right]$ and ${\gamma _{m}(t)} = \min
\left[ { \frac{{V\beta }}{{{H_{m}(t)}}},A_{m}^{\max}} \right] $,
where a larger ${H_{j}}(t)$ decreases ${\gamma _{j}}(t)$, which in
turn avoids the further increase of ${H _{j}}(t)$.
\subsubsection{Optimal Traffic Admission Control}

The optimal traffic admission control can be obtained by minimizing
the second item of R.H.S. of (\ref{drift}) at each slot, i.e. $
\sum\limits_{m \in {\mathcal {U}_H}} [{ {H_{m}}(t) - {Q_{m}}(t)]
{R_{m}}(t) } + {\sum\limits_{j \in {\mathcal {U}_R}} [{ {H_{j}}(t)
- {Q_{j}}(t)] {R_{j}}(t) }}$. Similarly, it can be further decoupled
to be computed for each UE separately as follows

\begin{equation}
\begin{array}{l}
\mathop {\max }\limits_{{R_{m}}} ~~ (
{H_{m}}(t) - {Q_{m}}(t)) {R_{m}}(t)\\
~~{\rm{s.t.}}~~ {R_{m}}(t) \le A_{m}(t),
\end{array}
\end{equation}
\begin{equation}
\begin{array}{l}
\mathop {\max }\limits_{{R_{j}}} ~~( {H_{j}}(t) -
{Q_{j}}(t)) {R_{j}}(t)\\
~~{\rm{s.t.}}~~ {R_{j}}(t) \le A_{j}(t),
\end{array}
\end{equation}
which are linear problems with the following optimal solutions:
\begin{equation}
{R_{m}}(t) = \left\{ \begin{array}{l}
 A_m(t), ~{\rm{if}}~ {H_{m}}(t) - {Q_{m}}(t) > 0, \\
 0,~~~~~~~ {\rm else}, \\
 \end{array}\right.\label{optimaladmissionm}
\end{equation}
\begin{equation}
{R_{j}}(t) = \left\{ \begin{array}{l}
 A_{j}(t),~{\rm if}~{H_{j}}(t) - {Q_{j}}(t) > 0, \\
 0,~~~~~~~ {\rm else}. \\
 \end{array}\right.\label{optimaladmissionij}
\end{equation}

This is a simple threshold-based admission control strategy. When
the traffic queue $Q_{m}(t)$ (or $Q_{j}(t)$) is smaller than a
threshold $H_{m}(t)$ (or $H_{j}(t)$), then the newly traffic
arrivals are admitted into the maintained traffic queues.
Consequently, this not only reduces the value of $H_{m}(t)$ (or
$H_{j}(t)$)so as to push $\gamma_{m}(t)$ (or $\gamma_{j}(t)$) to
become closer to $R_{m}(t)$ (or $R_{j}(t)$), but also increases the
throughput $R_{m}(t)$ (or $R_{j}(t)$) so as to improve the utility.
On the other hand, when traffic queue $Q_{m}(t)$ or ($Q_{j}(t)$) is
larger than a threshold $H_{m}(t)$ (or $H_{j}(t)$), then the traffic
arrivals will be denied to ensure the stability of traffic queues.

\subsubsection{Optimal User Association, RB and Power Allocation}

The optimal user association, RB, and power allocation at slot
$t$ can be obtained by minimizing the remaining item of R.H.S. of
(\ref{drift}), which is expressed as
\begin{equation}
\begin{array}{l}
\mathop {\min }\limits_{{\bf{s}},{\bf{p}},{\bf{a}}} -\!\!\!
\sum\limits_{m \in {\mathcal {U}_H}}\!\!\! {B_m(t){\mu _m}(t)}
-\!\!\! \sum\limits_{j \in {\mathcal {U}_R}} \!\!\!{B_j(t){\mu
_j}(t)}\\
~~~~~+ Y_R(t)\sum\limits_{i \in \mathcal {R}} {{p_i}(t)} + Y_H(t)p_H(t) \\
~{\rm s.t.~C1,C2,C3,C4,C8}. \\
\end{array}
\label{subproblem3}
\end{equation}
where $B_m(t) = Q_m(t)\tau + Z(t)$, $B_j(t) = Q_j(t)\tau + Z(t)$,
$Y_R(t) = W{\eta _{{\rm EE}}^{\rm req}\varphi _{\rm eff}^R }{Z}(t)$,
$Y_H(t) = W{\eta _{{\rm EE}}^{\rm req}\varphi _{\rm eff}^H }{Z}(t)$.
However, since the transmission rate $\mu_{m}(t)$, $\mu_{j}(t)$ and
the transmit power consumption $p_i(t)$ and $p_H(t)$ are functions
of user association $s_m(t)$, RB allocation $a_{jk}(t)$,
$a_{mk}(t)$ and $b_{ml}(t)$ and power allocation $p_{ijk}(t)$,
$p_{imk}(t)$ and $p_{ml}(t)$, this subproblem is a mixed-integer
nonconvex problem and is usually prohibitively difficult to solve.
To address this challenge, the computationally efficient algorithm
for this subproblem will be studied in the next section.

\section{Optimal User Association, RB and Power Allocation}

In this section, we commit to an effective method to solve the
subproblem of user association, RB and power allocation. The
continuity relaxation of binary variables and the Lagrange dual
decomposition method will be first utilized, upon which the optimal
primal solution is then obtained. As this subproblem is optimized at
each slot, the slot index $t$ will be ignored for brevity.

\subsection{Continuity Relaxation}
The multiplicative binary variables are first removed as ${x_{mk}} =
(1 - {s_m}){a_{mk}}$ and ${y_{ml}} = {(1 - s_m)}{b_{ml}} $, where
${x_{mk}} \in [0,1]$ and ${y_{ml}} \in [0,1]$. The binary variables
$a_{jk}$, $x_{mk}$ and $y_{ml}$ are then relaxed to take continuous
values in [0,1]. Furthermore, to make the problem tractable, the
auxiliary variables are introduced as ${w_{ijk}} =
{a_{jk}}{p_{ijk}}$, ${v_{imk}} = {x_{mk}}{p_{imk}}$ and {$u_{ml} =
y_{ml}p_{ml}$}. Let ${\bf{x}} = [a_{jk}, x_{mk}, y_{ml}: j \in
\mathcal {U}_R, m \in \mathcal {U}_H, k \in \mathcal {K}_R, l \in
\mathcal {K}_H]$ denote the vector of relaxed RB allocation
variables. Let ${\bf{w}} = [w_{ijk}, v_{imk}, u_{ml} : i \in
\mathcal {R}, j \in \mathcal {U}_R, m \in \mathcal {U}_H, k \in
\mathcal {K}_R, l \in \mathcal {K}_H]$ denote the vector of
introduced auxiliary variables. Thus the optimization problem
(\ref{subproblem3}) can be finally rewritten as
\begin{equation}
\begin{array}{l}
\mathop {\min }\limits_{{\bf{x}},{\bf{w}}} -\! \sum\limits_{m \in
{\mathcal {U} _H}} {{B_m}(\sum\limits_{l \in \mathcal {K}_H}
{{y_{ml}}{{\log }_2} (1 + {g_{ml}}{u_{ml}}/y_{ml})} }\\
~~~~~{+\sum\limits_{k \in \mathcal {K}_R} {{x_{mk}}{{\log }_2}
(1 + \sum\limits_{i \in \mathcal {R}} {{v_{imk}}{g_{imk}}} /{x_{mk}})} )} \\
~~~~~- \!\sum\limits_{j \in {\mathcal {U} _R}}
{\!\!{B_j}\!\!\sum\limits_{k \in \mathcal {K}_R} {{\!\!a_{jk}}{{\log
}_2}(1 + \sum\limits_{i \in \mathcal {R}} {{w_{ijk}}{g_{ijk}}}
/{a_{jk}})} }\\
~~~~~+ Y_R\sum\limits_{i \in \mathcal {R}} {(\sum\limits_{k \in
\mathcal {K}_R}\! {\sum\limits_{m \in {\mathcal {U} _H}} {{v_{imk}}}
} \!+\!\! \sum\limits_{k \in \mathcal
{K}_R}\! {\sum\limits_{j \in {\mathcal {U} _R}} {{w_{ijk}}} } )}\\
~~~~~+ Y_H \sum\limits_{m \in {\mathcal {U}_H}} {\sum\limits_{l \in
\mathcal {K}_H} {{u_{ml}}} }
\\
~{\rm s.t.}~\sum\limits_{j \in {\mathcal {U} _R}} {{a_{jk}}} +
\sum\limits_{m \in {\mathcal {U} _H}} {{x_{mk}}} \le 1,\forall k,\\
~~~~~\sum\limits_{m \in {\mathcal {U} _H}} {{x_{ml}}} \le 1,\forall l,\\
~~~~~~\sum\limits_{k \in \mathcal {K}_R} {\sum\limits_{m \in {\mathcal {U} _H}} {{v_{imk}}} } + \sum\limits_{k \in \mathcal {K}_R} {\sum\limits_{j \in {\mathcal {U} _R}} {{w_{ijk}}} } \le p_i^{\max },\forall i,\\
~~~~~~\sum\limits_{m \in {\mathcal {U}_H}} {\sum\limits_{l \in\mathcal{K}_H} {{u_{ml}}} } \le p_H^{\max},\\
~~~~~~~~a_{jk}, x_{mk}, y_{ml} \in [0,1], \forall j,k,m,l.
\end{array}
\label{relax}
\end{equation}

Since the term {$-{{x_{mk}}{{\log }_2}(1 + \sum\limits_{i \in
\mathcal {R}} {{v_{imk}}{g_{imk}}} /{x_{mk}})}$}, $ - {y_{ml}}{\log
_2}(1 + {g_{ml}}{u_{ml}}/{y_{ml}})$ and $-{{a_{jk}}{{\log }_2}(1\!\!
+\!\!\!\!\sum\limits_{i \in \mathcal {R}} \!\!\!{{w_{ijk}}{g_{ijk}}}
/{a_{jk}})}$ are the perspective functions of convex functions
$-{{{\log }_2}(1 + \sum\limits_{i \in \mathcal {R}}
{{v_{imk}}{g_{imk}}})}$, $ - {\log _2}(1 + {g_{ml}}{u_{ml}})$ and
$-{{{\log }_2}(1 + \sum\limits_{i \in \mathcal {R}}
{{w_{ijk}}{g_{ijk}}})}$, respectively, the objective of
(\ref{relax}) is a convex function. Furthermore, the constraints of
(\ref{relax}) are all linear with the continuity relaxation of
binary variables. According to the Salter's condition, the zero
Lagrange duality gap is guaranteed{\cite{cvx}}.

\subsection{Dual Decomposition}

The convex optimization problem can be solved by Lagrange dual
decomposition. Specifically, the Lagrangian function of the primal
objective function is given by
\begin{equation}
\begin{array}{l}
L({\bm{\lambda}}) = \mathop {\min }\limits_{{\bf{x}},{\bf{w}}}
-\!\!\! \sum\limits_{m \in {\mathcal {U} _H}}
\!\!\!{{B_m}(\sum\limits_{l \in
\mathcal {K}_H} {{y_{ml}}{{\log }_2}(1+u_{ml}g_{ml}/y_{ml})} }\\
{+\sum\limits_{k \in \mathcal {K}_R} {{x_{mk}}{{\log }_2}(1 +
\sum\limits_{i \in \mathcal {R}} {{v_{imk}}{g_{imk}}} /{x_{mk}})} )} \\
- \sum\limits_{j \in {\mathcal {U} _R}} {{B_j}\sum\limits_{k \in
\mathcal {K}_R} {{a_{jk}}{{\log }_2}(1 + \sum\limits_{i \in \mathcal
{R}} {{w_{ijk}}{g_{ijk}}} /{a_{jk}})} } \\
+ \sum\limits_{i \in\mathcal {R}} {({Y_R} + {\theta
_i})(\sum\limits_{k \in \mathcal {K}} {\sum\limits_{m \in {\mathcal
{U} _H}} {{v_{imk}}} } + \sum\limits_{k \in \mathcal {K}_R}
{\sum\limits_{j \in {\mathcal {U}
_R}} {{w_{ijk}}} } )}\\
- \sum\limits_{i \in \mathcal {R}} {{\theta _i}p_i^{\max }} + (Y_H +
\theta _0)\!\!\sum\limits_{m \in {\mathcal {U}_H}}\! {\sum\limits_{l
\in \mathcal
{K}_H} {{u_{ml}}} } -\theta _0p_H^{\max}, \\
\end{array}\label{LagDual}
\end{equation}
where ${\bm{\theta}} = [\theta _0, \theta _1, \theta _2, ..., \theta
_N]$ is the vector of Lagrangian dual variables related to the HPN
and RRH transmit power constraints. The Lagrangian dual function is
given by
\begin{equation}
\begin{array}{l}
D({\bm{\theta}}) = \mathop {\min
}\limits_{{\bf{x}},{\bf{w}}}L({\bm{\theta}})\\
~{\rm s.t.}\sum\limits_{m \in {\mathcal {U} _H}} {{y_{ml}}} \le 1, \forall l,\\
~~~~~~\sum\limits_{j \in {\mathcal {U} _R}} {{a_{jk}}} + \sum\limits_{m \in {\mathcal {U} _H}} {{x_{mk}}} \le 1, \forall k,\\
~~~~~~~~a_{jk}, x_{mk}, y_{ml} \in [0,1], \forall j,k,m,l.
\end{array}
\end{equation}
and the dual optimization problem is given by
\begin{equation}
\begin{array}{l}
{\mathop {\max }\limits_{\bm{\theta}} ~D(\bm{\theta} )} \\
{\rm s.t.}~~{\bm{\theta}} \succeq 0,\\
\end{array}
\end{equation}

Based on the Karush-Kuhn-Tucker (KKT) conditions, the optimal power
allocation can be obtained by differentiating the objective function
of (\ref{LagDual}) with respect to $v_{{imk}}$, $w_{{ijk}}$ and
$u_{ml}$, which are given by

\begin{equation}
v_{{imk}}^* \!=\! {\left[ {\frac{{{B_m}}}{{(Y_R + {\theta _i})\ln
2}} - \frac{{1 + \sum\limits_{i' \ne i} {v_{{i'mk}}^*{g_{i'mk}}}
}}{{{g_{imk}}}}} \right]^ + }{x_{mk}}, \label{zimk}
\end{equation}
\begin{equation}
w_{{ijk}}^* \!=\! {\left[ {\frac{{{B_j}}}{{(Y_R + {\theta _i})\ln
2}} - \frac{{1 + \sum\limits_{i' \ne i} {w_{{i'jk}}^*{g_{i'jk}}}
}}{{{g_{ijk}}}}} \right]^ + }{a_{jk}},\label{wijk}
\end{equation}
\begin{equation}
u_{ml}^* = {\left[ {\frac{{{B_{m}}}}{{({Y_H} + {\theta _0})\ln 2}} -
\frac{1}{{{g_{ml}}}}} \right]^ + }{y_{ml}},\label{uml}
\end{equation}
where $[x]^+ = \max\{x,0\}$. The derived power allocations have the
form of multi-level watering-filling and the water-filling levels
are determined by the traffic queue states and the virtual queue
state.

Substituting the optimal power allocations $v_{{imk}}^*$,
$w_{{ijk}}^*$ and $u_{ml}^*$ into (\ref{LagDual}) and denoting
\begin{equation}
{\Phi _{mk}}\! =\!\! \sum\limits_{i \in \mathcal {R}}\!\! {(Y_R +
{\theta _i}){p_{imk}}} \!\!-\!\! {B_m}\!{R_b}{\log _2}(1\! +\!\!
\sum\limits_{i \in \mathcal {R}} {{p_{imk}}{g_{imk}}} ),
\end{equation}
\begin{equation}
{\Lambda _{jk}} = \sum\limits_{i \in \mathcal {R}} {(Y_R + {\theta
_i}){p_{ijk}}} - {B_j}{R_b}{\log _2}(1 + \sum\limits_{i \in \mathcal
{R}} {{p_{ijk}}{g_{ijk}}} ),
\end{equation}
\begin{equation}
{\Gamma _{ml}} =(Y_H + \theta _0)p_{ml} - {B_{m}}{R_b}\log (1 +
{g_{ml}}{p_{ml}}),
\end{equation}

For notation simplicity, the dual function can be simplified as

\begin{equation}
\begin{array}{l}
\mathop {\min }\limits_{\bf {x}}\! \sum\limits_{m \in {\mathcal {U} _H}} \!
{\sum\limits_{k \in \mathcal {K}_R} {\!\!\!{\Phi _{mk}}{x_{mk}}} } + \!\!\!
\sum\limits_{m \in {\mathcal {U} _H}} \!{\sum\limits_{l \in \mathcal {K}_H}
 {\!\!\!{\Gamma _{ml}}{y_{ml}}} } + \!\!\!\sum\limits_{j \in {\mathcal {U} _R}}\!
 {\sum\limits_{k \in \mathcal {K}_R} {\!\!\!{\Lambda _{jk}}{a_{jk}}} } \\
~{\rm s.t.}\sum\limits_{m \in {\mathcal {U} _H}} {{y_{ml}}} \le 1, \forall l,\\
~~~~~\sum\limits_{j \in {\mathcal {U} _R}} {{a_{jk}}} + \sum\limits_{m \in {\mathcal {U} _H}} {{x_{mk}}} \le 1, \forall k,\\
~~~~~~a_{jk}, x_{mk}, y_{ml} \in [0,1], \forall j,k,m,l.
\end{array}
\end{equation}
{which is a linear programming (LP) problem. It can be proven that
if the bounded linear programming problem has an optimal solution,
then at least one of the optimal solutions is composed of the
extreme points \cite{LPbook}.

With the continuity relaxation,} the optimal RB allocation and
user association will be derived effectively according to the
following scheme.

\begin{itemize}
\item For the RB $k$ of RRH tier, the RB allocation
to HUE $m$ is decided by
\begin{equation}
x_{mk} = \left\{ {\begin{array}{*{20}{c}}
\begin{array}{l}
\!\!1,~{\rm if}~m = \arg \min \{ {\Phi _{mk}}:m \in {\mathcal {U}_{H}}\}\\
~~~~\& {\Phi _{mk}} < \min \{ {\Lambda _{jk}}:j \in {\mathcal {U}_R}\} \\
~~~~\& {\Phi _{mk}} < \min \{ {\Gamma _{ml}}:l \in {\mathcal {K}_H}\}, \\
\end{array} \\
\!\!{0,~{\rm else.}}~~~~~~~~~~~~~~~~~~~~~~~~~~~~~~~~~~ \\
\end{array}} \right.\label{OpYmkR}
\end{equation}
If there is RB of RRH tier allocated to HUE $m$, then we have $s_m
= 1$.

\item The remaining RBs of RRH tier will be allocated to RUEs.
Let $\mathcal {K}'_R$ denote the remaining RBs of RRH tier, then for
RB $k \in \mathcal {K}'_R$, the $a_{jk}$ is given by
\begin{equation}
{a_{jk}} = \left\{ {\begin{array}{*{20}{c}}
{1,{~\rm if~}j = \arg \min \{ {\Lambda _{jk}}:j \in {\mathcal {U}_R}\} \& {\Lambda _{jk}} < 0}, \\
{0,~{\rm else.}} ~~~~~~~~~~~~~~~~~~~~~~~~~~~~~~~~~~~~~~~~~~ \\
\end{array}} \right.\label{OpAijk}
\end{equation}

\item After the RB allocation of RRH tier is accomplished, the RBs of
HPN tier will be allocated. Let $\mathcal {U}_0'$ denote the set of
HUEs that are served by HPN. The RB allocation $y_{ml}$ is given by
\begin{equation}
y_{ml} = \left\{ {\begin{array}{*{20}{c}}
{1,{~\rm if~}m = \arg \min \{ {\Gamma _{ml}}:m \in {\mathcal {U}_0'}\} ,} \\
{0,~{\rm else.}~~~~~~~~~~~~~~~~~~~~~~~~~~~~~~~~~} \\
\end{array}} \right.\label{OpYmkH}
\end{equation}
\end{itemize}

{It is worth noting that, after the continuity relaxation, the
binary $x_{mk}$, $a_{jk}$ and $y_{ml}$ can be still obtained at the
extreme point of constraint set, i.e., 0 or 1.}

To recover the optimal primal solution, the dual variables are then
iteratively computed using the subgradient
method{\cite{subgrident}},
\begin{equation}
\theta _0^{(n + 1)} = {\left[ {\theta _0^{(n)} + {\xi
_0^{(n+1)}}\nabla _0^{(n + 1)}} \right]^ + },\label{Lagdualupdate1}
\end{equation}
\begin{equation}
\theta _i^{(n + 1)} = {\left[ {\theta _i^{(n)} + \xi _i^{(n +
1)}\nabla _i^{(n + 1)}} \right]^ + },\label{Lagdualupdate2}
\end{equation}
where $n$ is the iteration index, $\xi _0^{(n)}$ and $\xi _i^{(n)}$
is the step size at the $n$-th iteration to guarantee the
convergence, $\nabla _0^{(n + 1)} $ and $\nabla _i^{(n + 1)}$ are
the subgradient of the dual function, which are given by
\begin{equation}
\nabla _0^{(n + 1)} = \left( {\sum\limits_{m \in {\mathcal {U}_H}}
{\sum\limits_{l \in \mathcal {K}_H} {{u_{ml}^{(n)}}} } - p_H^{\max
}} \right),
\end{equation}
\begin{equation}
\nabla _i^{(n + 1)} \!=\! \left( {\sum\limits_{j \in {\mathcal
{U}_R}}\! {\sum\limits_{k \in \mathcal {K}_R} \!\!{w_{ijk}^{(n)}} }
+\!\! \sum\limits_{j \in {\mathcal {U}_R}}\! {\sum\limits_{k \in
\mathcal {K}_R} \!\!{v_{imk}^{(n)}} }\! - p_i^{\max }} \right).
\end{equation}

Finally, the overall procedure of joint congestion control and
resource optimization is summarized in the \textbf{Algorithm 1}.

\begin{algorithm}[h]
\begin{algorithmic}[1]
\caption{The Joint Congestion Control and Resource Optimization
Algorithm} \STATE For each slot, observe the traffic queues
$Q_m(t)$, $Q_{j}(t)$ and the virtual queues $H_m(t)$, $H_{j}(t)$,
$Z(t)$; \STATE Calculate the optimal auxiliary variables $\gamma
_m(t)$ and $\gamma _{j}(t)$ by solving (\ref{optimalauxij}) and
(\ref{optimalauxm}); \STATE Determine the optimal amount of admitted
traffics $R_m(t)$ and $R_{j}(t)$ according to
(\ref{optimaladmissionm}) and (\ref{optimaladmissionij}); \REPEAT
\STATE Obtain the optimal power allocation $p_{imk}$ and $p_{ijk}$
of RRH tier by iteratively updating (\ref{zimk}) and
(\ref{wijk});\STATE Obtain the optimal power allocation $p_{ml}$ of
HPN according to (\ref{uml}); \STATE Obtain the optimal RB
allocation $a_{mk}$ and $a_{jk}$ of RRH tier according to
(\ref{OpYmkR}) and (\ref{OpAijk}) and derive the optimal user
association $s_m$; \STATE Obtain the optimal RB allocation $b_{ml}$
of HPN tier according to (\ref{OpYmkH});\STATE Update the Lagrangian
dual variables $\bm{\theta}$ according to (\ref{Lagdualupdate1}) and
(\ref{Lagdualupdate2});\UNTIL certain stopping criteria is met;
\STATE Update the traffic queues $Q_m(t)$, $Q_{j}(t)$ and the
virtual queues $H_m(t)$, $H_{j}(t)$ and $Z(t)$ according to
(\ref{dynamicsQ0m}), (\ref{dynamicsQij}), (\ref{vq0m}), (\ref{vqij})
and (\ref{vqpower}).
\end{algorithmic}
\end{algorithm}

\section{Performance Bounds}

In this section, the performance bounds of the proposed
algorithm based on Lyapunov optimization will be mathematically
analyzed.

\subsection{Bounded Queues}
Suppose $\phi_H $ and $\phi_R$ are the largest right-derivative of
$g_H(.)$ and $g_R(.)$, respectively, then the proposed algorithm
based on Lyapunov optimization ensures that the traffic queues are
bounded, which is given by \textbf{Theorem} \ref{TheoremQ}.

\begin{theorem}
\emph{For arbitrary traffic arrival rates (possibly exceeding the
network capacity of H-CRAN) {and certain EE requirement}, an H-CRAN
using the proposed algorithm with any $V \ge 0$ can guarantee the
following bounds of traffic queues:}
\begin{equation}
{Q_{j}}(t) \le V\alpha {\phi _R} + 2{A_{j}^{\max }},\label{boundQij}
\end{equation}
\begin{equation}
{Q_{m}}(t) \le V\beta {\phi _H} + 2{A_{m}^{\max}}.\label{boundQ0m}
\end{equation}
\label{TheoremQ}
\end{theorem}

\begin{proof}
Please refer to Appendix D.
\end{proof}

\subsection{Utility Performance}
The utility performance of proposed solution based on Lyapunov
optimization is given by \textbf{Theorem} \ref{TheoremU}.

\begin{theorem}
\emph{For arbitrary arrival rates{ and certain EE requirement}, an
H-CRAN using the proposed algorithm with any $V \ge 0$ can provides
the following utility performance under certain EE requirement:
\begin{equation}
U(\bar {\bf{r}}) \ge {U^*} - C/V,\label{utilitybound}
\end{equation}
where ${U^*}$ is the optimal infinite horizon utility over all
algorithms that stabilize traffic queues {and satisfy required EE
performance constraint.}}\label{TheoremU}
\end{theorem}

\begin{proof}
Please refer to Appendix E.
\end{proof}

{To readily understand the obtained results indicated in
\textbf{Theorem} 3 and \textbf{Theorem} 4, some important
observations are further provided as follows.

\begin{itemize}
\item \textbf{Theorem} 4 shows that $U(\bar {\bf{r}}) \ge {U^*} -
C/V$. Besides, $U(\bar {\bf{r}}) \le {U^*}$. Therefore, we have
${U^*} - C/V \le U(\bar {\bf{r}}) \le {U^*}$, which indicates that
$U(\bar {\bf{r}})$ can be arbitrarily close to ${U^*}$ by setting
a sufficiently large $V$ to make $C/V$ arbitrarily small and close to 0.
This will be further verified in the following simulation section as shown in Fig. 2.

\item \textbf{Theorem} 3 and \textbf{Theorem} 4
both show the delay-utility tradeoff of $[\mathcal {O}(V),
1-\mathcal {O}(1/V)]$, which provides an important guideline to explicitly balance the delay-throughput performance on demand.
This will also be further verified in the simulation section as shown in Fig. 2 and Fig. 3.
\end{itemize}
}

\begin{remark}
{The traffic models are not specified throughout this paper, as they
do not affect the problem formulation and the corresponding
analysis. Moreover, although the packet traffic arrivals with i.i.d.
and constant arrival rates are considered in this paper, the
proposal and the corresponding theoretical analysis results still
hold for other arrivals that are independent from slot to slot, but
their arrival rates are time-varying and ergodic (possibly
non-i.i.d.). The reason is that the joint congestion control and
resource optimization policy is made only based on the size of the
queues without requiring the knowledge of traffic arrivals.
Therefore, the proposal is robust to the traffic arrival
distribution model.}
\end{remark}

\section{Simulations}

In this section, simulations will be carried out to evaluate the
performances of proposed
Joint-Congestion-Control-and-Resource-Optimization (JCCRO) scheme in
an H-CRAN.

\subsection{Parameters Setting}

The considered H-CRAN consists of 1 HPN, 4 RRHs, 12 HUEs and 10
RUEs. The HPN is located in the center of the cell area, while the
RRHs, HUEs and RUEs are uniformly distributed. There are 8 RBs and
12 RBs in the RB sets $\mathcal {K}_H$ and $\mathcal {K}_R$,
respectively. The bandwidth of each RB is $W_0 = 15$ kHz, so the
system bandwidth is $W = 300$ kHz. The slot duration is 0.01 second.
The path loss model of RRH and HPN is given by $31.5 +
40.0\log10(d)$ and $31.5 + 35.0\log10(d)$, respectively, where $d$
denotes the distance between transmitter and receiver in meters. The
fast-fading coefficients are all generated as i.i.d. Rayleigh random
variables with unit variances. The noise power is -102 dBm. For the
HPN, the drain efficiency, the maximum transmit power consumption
and the static power consumption are given by $\varphi _{\rm eff}^H
= 1$, $p^{\max}_H = 10$ W, $p_{\rm c}^H = 2$ W, respectively. For
the RRHs, the drain efficiency, the total transmit power consumption
and the static power consumption are given by $\varphi _{\rm eff}^R
= 1$, $p_i^{\max} = 3$ W, $p_{\rm c}^R = 1$ W, respectively. For
simplicity of comparison, the utility function of total average
throughput is adopted, i.e. $U({\bar {\bf{r}}}) = \alpha
{\sum\limits_{j \in {\mathcal {U}_R}} {{{\bar r_{j}}} } } + \beta
\sum\limits_{m \in \mathcal {U}_H} {{{\bar r_{m}}} }$, where the
positive utility prices for RUEs and HUEs are $\alpha = 1$ and
$\beta = 1$, respectively. It is worth noting that in this special
case, the first subproblem to derive optimal auxiliary variables is
not required. The traffic arrivals of HUEs and RUEs follow Poisson
distribution, and the mean traffic arrival rate for RUE $\lambda _j$
and HUE $\lambda _m$ is given by $\lambda _j = \lambda$ and $\lambda
_m = 0.5\lambda$, respectively. Each point of the following curves
is averaged over 5000 slots.

\subsection{The Delay-Throughput Tradeoff with Guaranteed EE}

Fig. \ref{Fig_V_Thr} and Fig. \ref{Fig_V_Delay} illustrate the
performances of throughput, delay with guaranteed EE versus
different control parameter $V$ when the mean traffic arrival rate
is $\lambda = 6$ kbits/slot. As can be seen, the achieved utility of
total average throughput increases to optimum at the speed of
$\mathcal {O}(1/V)$ as $V$ increases, which is due to the fact that
a larger $V$ implies that the control solution emphasizes more on
throughput utility. However, the utility improvement starts to
diminish with excessive increase of $V$, which can adversely
aggravate the congestion as the average delay increases linearly
with $V$. All these verify the observations indicated by
\textbf{Theorem} \ref{TheoremQ} and \textbf{Theorem} \ref{TheoremU}.
Furthermore, Fig. \ref{Fig_V_EE} plots the achieved EE verses
different control parameter $V$, which shows that the achieved EE is
always larger than or equal to $\eta_{\rm EE}^{\rm req}$.
\begin{figure}
\centering
\includegraphics[scale=0.65]{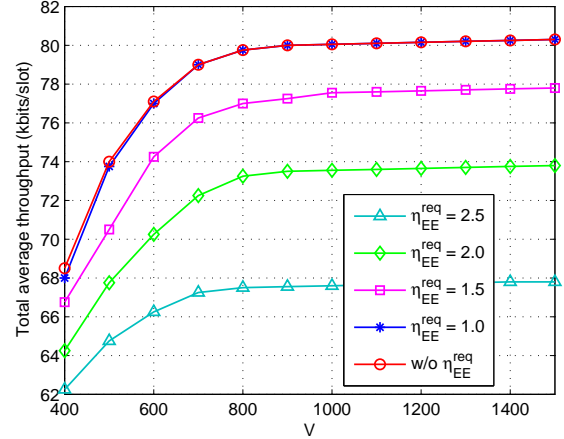}
\caption{Total average throughput versus control parameter
$V$}\label{Fig_V_Thr}
\end{figure}
\begin{figure}
\centering
\includegraphics[scale=0.65]{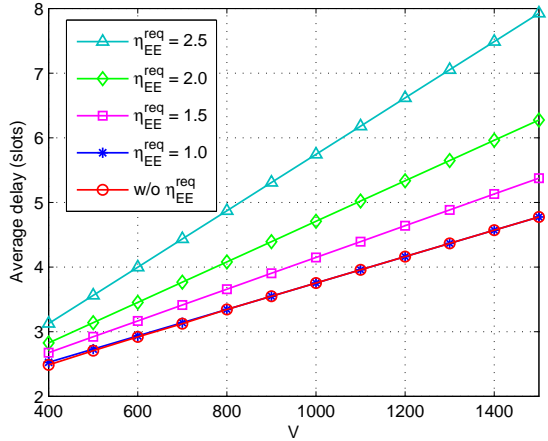}
\caption{Average delay versus control parameter
$V$}\label{Fig_V_Delay}
\end{figure}
\begin{figure}
\centering
\includegraphics[scale=0.65]{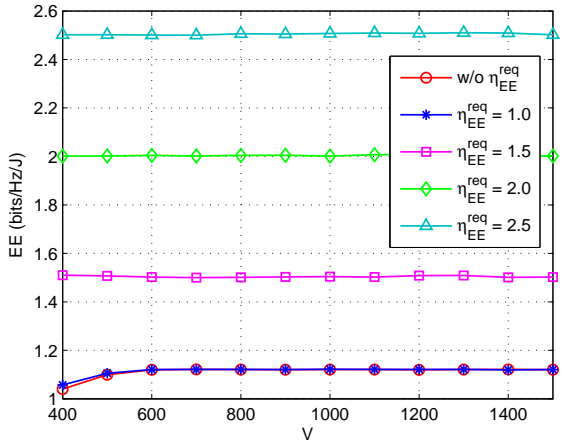}
\caption{Achieved EE versus control parameter $V$}\label{Fig_V_EE}
\end{figure}

It can be further observed from Fig. \ref{Fig_V_Thr} - Fig.
\ref{Fig_V_EE} that there exists a ceratin EE threshold $\eta_{\rm
EE}^{\rm thr}$ of the network when making a tradeoff between delay
and throughput. In our simulations, the EE threshold is $\eta_{\rm
EE}^{\rm thr} = 1.12$, which is actually the EE archived by the case
without EE requirement. Specifically, when the required EE is below
$\eta_{\rm EE}^{\rm thr}$, the actually achieved EE and
delay-throughput tradeoff is almost the same as the situation
without EE requirement. Once the required EE is above the threshold
$\eta_{\rm EE}^{\rm thr}$, the total average throughput sharply
decreases (see Fig. \ref{Fig_V_Thr}), and the average delay also
increases (see Fig. \ref{Fig_V_Delay}). This is because, to
guarantee the required EE, the network has to decrease the transmit
power, which further result in the decrease of transmit rate,
followed by the decrease of achieved throughput and the increase of
average delay. All the above observations indicate that the network
can guarantee the EE performance when maximizing the throughput. At
this point, the control parameter $V$ provides a controllable method
to flexibly balance throughput-delay performance tradeoff with
guaranteed EE. To let the H-CRAN work in a preferred state, what we
only need to do is to select an appropriate control parameters $V$.

{

\subsection{The Convergence of The Proposed Solution}

Fig. \ref{Fig_V_Iteration} shows the average number of convergence
iterations for the proposal. It can be generally observed that the
proposal under different EE requirements can converge fairly fast.
Besides, the convergence speed is influenced by some key parameters.
On the one hand, a larger $V$ means a larger average sum rate and
then a slower convergence. On the other hand, as clarified in Fig.
\ref{Fig_V_Thr}, a larger EE requirement $\eta_{\rm{EE}}^{\rm{req}}$
makes a smaller average sum rate, which means a faster convergence.

\begin{figure}
\centering
\includegraphics[scale=0.75]{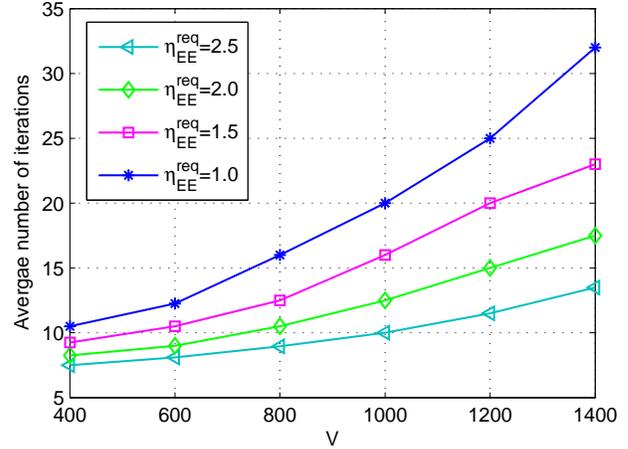}
\caption{Average number of convergence iterations versus control parameter
$V$}\label{Fig_V_Iteration}
\end{figure}

}

\subsection{The Performance Comparison under Different Traffic Arrival Rate}
To validate the efficacy of the proposed JCCRO scheme, we compare
its performances with the Maximum-Sum-Rate (MSR) scheme, which is
modeled as
\begin{equation}
\begin{array}{l}
\max \sum\limits_{j \in {\mathcal {U}_R}} {{\mu _j}(t)} + \sum\limits_{m \in {\mathcal {U}_H}} {{\mu _m}(t)} \\
~{\rm s.t.~ ~C1-C5, C8}.\\
\end{array}
\end{equation}

For the JCCRO scheme, we set the control parameter as $V = 1000$.
From Fig. \ref{Fig_Rate_Thr}, it can be observed that the total
average transmit rate of the proposed JCCRO scheme with different EE
requirement are the same and not less than the total traffic arrival
rate at first, then go to the maximum values as the mean traffic
arrival rate increases. From Fig. \ref{Fig_Rate_Delay}, it can be
observed that the average delay of the proposed JCCRO scheme always
increases with increasing mean arrival rate, that is because more
traffic arrivals means larger transmit rate, which cannot be large
enough due to the EE constraint. Again, the results of Fig.
\ref{Fig_Rate_Thr} and Fig. \ref{Fig_Rate_Delay} confirm that the
setting of EE requirement have a great effect on system performance.

As for the compared MSR scheme, on the one hand, the total average
transmit rate keeps unchanged as the traffic arrival rate varies
(see Fig. \ref{Fig_Rate_Thr}). The reason is that the MSR scheme
does not consider stochastic traffic arrivals and delivers data
under the full buffer assumption. On the other hand, the average
delay of MSR scheme is almost the same as that of JCCRO scheme at
first, but it begins to sharply increase to infinity as time elapse
when the arrival rate is large than a certain value(see. Fig.
\ref{Fig_Rate_Delay}). This is because both schemes are able to
timely transmit all the arrived data when the arrival rate is small,
while the traffic admission control component of JCCRO starts to
work to make the queues stable as the arrival rates increase.

In Fig. \ref{Fig_Rate_Power}, we further compare the total avergae
power consumption of the JCCRO scheme and the MSR scheme. We can see
that the power consumption of MSR scheme keeps unchanged as traffic
arrival rate varies and is much more than that of JCCRO scheme in
the relatively light traffic states. This is because that the MSR
scheme delivers data under the full buffer assumption and fails to
adapt to the traffic arrivals, which thus leads to a waste of energy
despite achieving the same EE performance. All the observations from
Fig. \ref{Fig_Rate_Thr} - Fig. \ref{Fig_Rate_Power} validate the
advantages of joint congestion control and resource optimization: 1)
in the relative light traffic states, more energy can be saved with
the adaptive resource optimization, and 2) in the relative heavy
traffic states, the traffic queues can be stabilized with the
traffic admission control.
\begin{figure}
\centering
\includegraphics[scale=0.65]{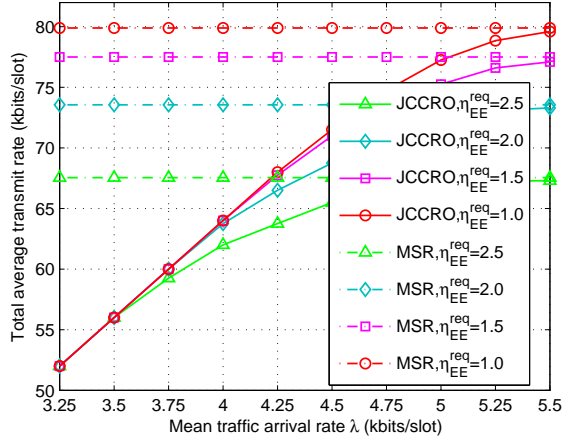}
\caption{Total average transmit rate versus mean traffic arrival
rate $\lambda$} \label{Fig_Rate_Thr}
\end{figure}
\begin{figure}
\centering
\includegraphics[scale=0.65]{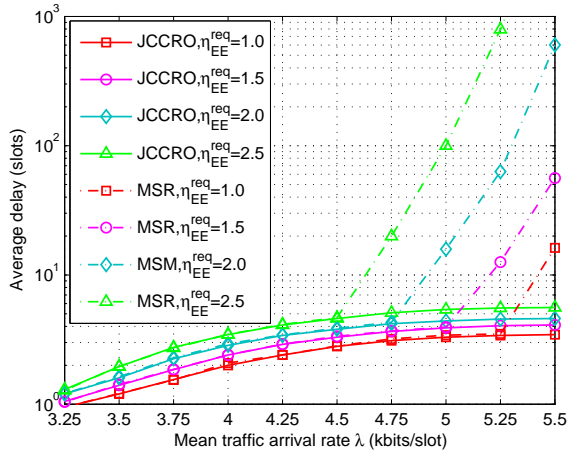}
\caption{Average delay versus mean traffic arrival rate $\lambda$}
\label{Fig_Rate_Delay}
\end{figure}
\begin{figure}
\centering
\includegraphics[scale=0.65]{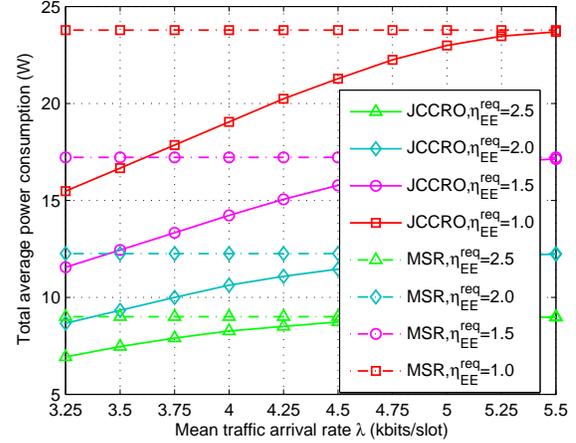}
\caption{Total average power consumption versus mean traffic arrival
rate $\lambda$} \label{Fig_Rate_Power}
\end{figure}

\section{Conclusion}

This work has focused on the stochastic optimization of
EE-guaranteed joint congestion control and resource optimization in
a downlink slotted H-CRAN. Based on the Lyapunov optimization
technique, this stochastic optimization problem has been transformed
and decomposed into three subproblems which are solved at each slot.
The continuality relaxation of binary variables and Lagrange dual
decomposition method have been exploited to solve the third
subproblem efficiently. An EE-guaranteed $[ \mathcal {O}(1/V ),
\mathcal {O}(V )]$ throughput-delay tradeoff has been finally
achieved by the proposed scheme, which has been verified by both the
mathematical analysis and numerical simulations. The simulation
results have shown the significant impact of EE requirement on the
achieved throughput-delay tradeoff and have validated the
significant advantages of joint congestion control and resource
optimization. For the future work, it would be interesting to extend
our proposed model to provide deterministic delay guarantee for
real-time traffic applications in realistic networks, e.g., mobile
video and voice.

\appendices
\section{Proof of Theorem 1}
Let $U _1^*$ and $U _2^*$ be the optimal utility of problems
(\ref{StochPro}) and (\ref{EqvStochPro}), respectively. For ease of
notation, let $\Omega _1^* $ and $\Omega _2^*$ be the optimal
solutions that achieve $U _1^*$ and $U _2^*$, respectively. Since
$U(.)$ is a non-decreasing concave function, by Jensen's inequality,
we have
\begin{equation}
U(\bar {\bm{\gamma}}) \ge \bar U({\bm{\gamma}}) = U_2^*.
\end{equation}

Since the solution $\Omega _2^*$ satisfies the constraint C10, then
we have
\begin{equation}
U(\bar {\bf{r}}) \ge U(\bar {\bm{\gamma}}).
\end{equation}

Furthermore, since $\Omega _2^*$ is feasible for the transformed
problem (\ref{EqvStochPro}), it also satisfies the constraints of
the original problem (\ref{StochPro}). Therefore, we can have
\begin{equation}
U _1^* \ge U(\bar {\bf{r}}) \ge U _2^*.
\end{equation}

Now we prove that $U _2^* \ge U _1^*$. Since $\Omega _1^*$ is an
optimal solution to the original problem, it satisfies the
constraints C1-C8, which are also the constraints of the transformed
problem. By choosing $\gamma_{m} = {\bar r_{m}^*}$ and $\gamma_{j} =
{\bar r_{j}^*}$ for all slot $t$ together with the policy $\Omega
_1^*$, we then have a feasible policy for the transformed problem
(\ref{EqvStochPro}), that is
\begin{equation}
U _2^* \ge \bar U({\bm{\gamma}}) = U(\bar {\bf{r}}) = U _1^*.
\end{equation}

Therefore, we have $U _1^* = U _2^*$ and can further conclude the
\textbf{Theorem} 1.

\section{Proof of Theorem 2}
The constraint C5 is proved firstly, and C10 can be proved
similarly. When the virtual queue $H_j(t)$ is stable, then we have
$\mathop {\lim }\limits_{T \to \infty } \frac{{\mathbb{E}[ H_{j}(T)]
}}{T} = 0$ with the probability 1. It is clear that ${H_{j}}(t + 1)
\ge {H_{j}}(t) - {R_{j}}(t) + {\gamma_{j}}(t)$. Summing this
inequality over time slots $t \in \{0,1,..., T-1 \}$ and dividing
the result by $T$ yields
\begin{equation}
\frac{{{H_{j}}(T) - {H_{j}}(0)}}{T} + \frac{1}{T}\sum\limits_{t =
0}^{T-1} {{R _{j}}(t)} \ge \frac{1}{T}\sum\limits_{t = 0}^{T-1}
{{\gamma_{j}}(t)}.
\end{equation}
By taking $T$ asymptotically closed to infinity, we finally have C5.

It is similarly concluded that we can have the inequality
$W\eta_{\rm EE}^{\rm req}\bar p_{\rm sum} \le \bar \mu _{\rm sum}$,
i.e. $\eta _{\rm EE} \ge \eta_{\rm EE}^{\rm req}$, only when the
virtual queue $Z(t)$ is stable.

\section{Proof of Lemma 1}

By leveraging the fact that ${(\max [a - b,0] + c)^2} \le {a^2} +
{b^2} + {c^2} - 2a(b - c),\forall a,b,c \ge 0$ and squaring Eq.
(\ref{dynamicsQ0m}), Eq. (\ref{dynamicsQij}), Eq. (\ref{vq0m}),
Eq. (\ref{vqij}) and Eq. (\ref{vqpower}), we have
\begin{equation}
Q_{m}^2\!(t + 1) - Q_{m}^2\!(t) \!\le\! R_{m}^2\!(t) \!+ \mu
_{m}^2\!(t)\tau ^2 \!\!- 2{Q_{m}}\!(t)({\mu _{m}}\!(t)\tau \!-
\!{R_{m}}\!(t)),
\end{equation}
\begin{equation}
Q_{j}^2(t + 1) - Q_{j}^2(t) \le R_{j}^2(t) + \mu _{j}^2(t)\tau ^2 -
2{Q_{j}}(t)({\mu _{j}}(t)\tau - {R_{j}}(t)),
\end{equation}
\begin{equation}
H_{m}^2\!(t + 1) - H_{m}^2\!(t) \le \gamma_{m}^2\!(t) + R
_{m}^2\!(t) - 2{H_{m}}\!(t)({R _{m}}\!(t) - {\gamma_{m}}\!(t)),
\end{equation}
\begin{equation}
H_{j}^2(t + 1) - H_{j}^2(t) \le \gamma_{j}^2(t) + R _{j}^2(t) -
2{H_{j}}(t)({R _{j}}(t) - {\gamma_{j}}(t)),
\end{equation}
\begin{equation}
\begin{array}{l}
Z^2(t + 1) - Z^2(t) \le (W\eta_{\rm EE}^{\rm req}{p_{\rm sum}}(t))^2
+ \mu_{\rm sum}^2(t)
-2Z(t)\\
~~~~~~~~~~~~~~~~~~~~~~~~~({\mu _{\rm sum}}(t) - W\eta_{\rm EE}^{\rm
req}{p_{\rm sum}}(t)).
\end{array}
\end{equation}

According to the definition of Lyapunov drift, we then have the
following expression by summing up the above inequalities and taking
expectation over both sides,
\begin{equation}
\begin{array}{l}
\mathbb{E}[ L(\chi (t + 1) - L(\chi (t))]
\le \\
\frac{1}{2}\!\!{\sum\limits_{j \in {\mathcal {U}_R}} \!\!\!\!{\mathbb{E}[ 2R_{j}^2(t) \!+\! \mu _{j}^2(t)\tau^2\!\! + \!\gamma _{j}^2] } } \!+ \!\frac{1}{2}\!\!\!\!\sum\limits_{m \in {\mathcal {U}_H}}\!\!\!\!\!{\mathbb{E}[ 2R_{m}^2(t)\! +\! \mu _{m}^2(t)\tau \!+\! \gamma _{m}^2] } \\
+ {\mathbb{E}[ W^2(\eta_{\rm EE}^{\rm req})^2\!p_{\rm sum}^2(t) \!\!+\!\! \mu_{\rm sum}^2(t)] }\! -\!\!\!\!\! {\sum\limits_{j \in {\mathcal {U}_R}} \!\!\!\!{\mathbb{E}[ {Q_{j}}(t)({\mu _{j}}(t)\tau \!\!- \!{R_{j}}(t))] } } \\
-\!\!\!\! \sum\limits_{m \in {\mathcal {U}_H}} \!\!\!{\mathbb{E}[ {Q_{m}}(t)({\mu _{m}}(t) \!-\! {R_{m}}(t))] }\! -\!\!\! {\sum\limits_{j \in {\mathcal {U}_R}}\!\!\! {\mathbb{E}[ {H_{j}}(t)({R_{j}}(t) \!-\! {\gamma _{j}}(t))] } } \\
-\!\!\!\!\! \sum\limits_{m \in {\mathcal {U}_H}} \!\!\!\!\!{\mathbb{E}[ {H_{m}}\!(t)({R_{m}}\!(t) \!-\!\! {\gamma _{m}}\!(t))] } \!-\! {\mathbb{E}[\! {Z}\!(t)({\mu_{\rm sum}\!(t)} \!\!-\!\! W\!\!\eta_{\rm EE}^{\rm req}\!{p_{\rm sum}}\!(t))] }, \\
\end{array}\label{approv2}
\end{equation}

Finally, the upper bound of drift-minus-utility expression can be
obtained as Eq. (\ref{drift}) by subtracting the expression
$V\mathbb{E}\{U({\bm{\gamma}})\}$ from the both sides of Eq.
(\ref{approv2}).

\section{Proof of Theorem 3}
The bounds of traffic queues for RUEs are proved firstly, and that
for MUEs can be proved similarly. Suppose that the following
inequality holds at slot $t$,
\begin{equation}
{H_{j}}(t) \le V\alpha {\phi _R} + {A_{j}^{\max}},\label{Hjbound}
\end{equation}

If $H_{j}(t) \le V\alpha {\phi _R}$, then it is easy to get
${H_{j}}(t) \le V\alpha {\phi _R} + {A_{j}^{\max}}$ according to the
admission constraint $R_{j}(t) \le A_{j}^{\max}$. Else if $H_{j}(t)
\ge V\alpha {\phi _R}$, since the utility function $g_R(.)$ is a
non-decreasing concave function and ${\phi _R}$ is the largest
right-derivative of $g_R(.)$, the following inequality can be easily
established,
\begin{equation}
\begin{array}{l}
V\alpha {g_R}({\gamma _{j}}(t)) - {H_{j}}(t){\gamma _{j}}(t) \le
V\alpha {g_R}(0)\\ + (V\alpha {\phi _R} - {H_{j}(t)}){\gamma
_{j}(t)} \le V\alpha {g_R}(0).
\end{array}
\end{equation}
which follows that when $H_{j}(t) \ge V\alpha {\phi _R}$, the
auxiliary variables decision in (\ref{optimalauxij}) forces $\gamma
_{j}$ to be 0. Therefore, inequality (\ref{Hjbound}) also holds at
slot $t+1$,
\begin{equation}
{H_{j}}(t+1) \le {H_{j}}(t) \le V\alpha {\phi _R} + {A_{j}^{\max }}.
\end{equation}

With above bound of virtual queue, the bound of traffic queue is
proved next. If $Q_{j}(t) \le H_{j}(t)$, according to the admission
control policy in (\ref{optimaladmissionij}), we have
\begin{equation}
\begin{array}{l}
{Q_{j}}(t + 1) = {Q_{j}}(t) + {R_{j}}(t) \le {Q_{j}}(t) +
{A_{j}^{\max}}\\ \le {H_{j}}(t) + {A_{j}^{\max}} = V\alpha {\phi _R}
+ 2{A_{j}^{\max}}.
\end{array}
\end{equation}

\section{Proof of Theorem 4}

To prove the bound of utility performance, the following lemma is
required.

\begin{lemma}
For arbitrary arrival rates, there exists a randomized stationary
control policy $\pi$ for H-CRAN that chooses feasible control
decisions independent of current traffic queues and virtual queues,
which yields the following steady state values:
\begin{equation}
\gamma _{m}^\pi(t) = r_{m}^*, \gamma _{j}^\pi(t) =
r_{j}^*,\label{randompolicy1}
\end{equation}
\begin{equation}
\mathbb{E}[R_{m}^\pi(t)] = r_{m}^*, \mathbb{E}[R_{j}^\pi(t)] =
r_{j}^*,\label{randompolicy1}
\end{equation}
\begin{equation}
\mathbb{E}[\mu_{m}^\pi(t)\tau] \ge \mathbb{E}[R_{m}^\pi(t)],
\mathbb{E}[\mu_{j}^\pi(t)\tau] \ge
\mathbb{E}[R_{j}^\pi(t)],\label{randompolicy2}
\end{equation}
\begin{equation}
\mathbb{E}[\mu _{\rm sum}^\pi(t)] \ge W\eta_{\rm EE}^{\rm
req}\mathbb{E}[{p_{\rm sum}^\pi}(t)].\label{randompolicy2}
\end{equation}
\label{lemma2}
\end{lemma}

As the similar proof of \textbf{Lemma} \ref{lemma2} can be found
in{\cite{neelybook}}, the details are omitted to avoid redundancy.
Since the proposed solution is obtained by choosing control
variables that can minimize the R.H.S. of Eq. (\ref{drift}) among
all feasible decisions (including the randomized control decision
$\pi$ in \textbf{Lemma} \ref{lemma2}) at each slot, then we have
\begin{equation}
\begin{array}{l}
\Delta (\chi (t)) \le C - \mathbb{E}\left[{ {\sum\limits_{j \in
{\mathcal {U}_R}} {( V\alpha {g_R}({\gamma _{j}^{\pi}(t)}) -
{H_{j}}(t){\gamma _{j}^\pi}(t)) } }}\right.\\
~~~~~~~~~~\left.{+ \sum\limits_{m \in {\mathcal {U}_H}} {( V\beta {g_H}({\gamma _{m}^\pi(t)})- {H_{m}}(t){\gamma _{m}^\pi(t)}} ) |\chi(t)}\right]\\
~~~~~~~~~~- \mathbb{E}\left[{ \sum\limits_{m \in {\mathcal {U}_H}}(
{{H_{m}}(t) - {Q_{m}}(t)) {R_{m}^\pi}(t) }}\right.\\
~~~~~~~~~~\left.{+ {\sum\limits_{j \in {\mathcal {U}_R}} {( {H_{j}}(t) - {Q_{j}}(t)] {R_{j}^\pi}(t)) \chi (t) } }}\right] \\
~~~~~~~~~~- \mathbb{E}\left[{ \sum\limits_{m \in {\mathcal {U}_H}} {{Q_{m}}(t)\mu_{m}^\pi(t)\tau }+ {\sum\limits_{j \in {\mathcal {U}_R}} {{Q_{j}}(t){\mu _{j}^\pi}(t)\tau} }}\right.\\
~~~~~~~~~~\left.{ + {{Z}(t)( \mu _{\rm sum}^\pi(t) - W\eta_{\rm EE}^{\rm req}{p_{\rm sum}^\pi}(t)) } |\chi (t)}\right], \\
\end{array}
\end{equation}

Since the randomized stationary policy is independent of $\chi (t)$,
we have
\begin{equation}
\begin{array}{l}
\Delta (\chi (t)) \le C - {\sum\limits_{j \in {\mathcal {U}_R}}[ \mathbb{E}[{ V\alpha {g_R}({\gamma _{j}^{\pi}(t)})] - {H_{j}}(t)\mathbb{E}[{\gamma _{j}^\pi}(t)] }}] \\
~~~~~~~~~+ \!\!\!\!\!\sum\limits_{m \in {\mathcal {U}_H}}\!\!\!\![\mathbb{E} {[ V\beta {g_H}({\gamma _{m}^\pi(t)})]\!\! - \!\!{H_{m}}(t)\mathbb{E}[{\gamma _{m}^\pi(t)}} ] ]\!- \!\!\!\!\!\!\sum\limits_{m \in {\mathcal {U}_H}}\!\!\!\! ( {H_{m}}\!(t)\\
~~~~~~~~~- {Q_{m}}(t)) \mathbb{E}[{R_{m}^\pi}(t) ]- \!\!\!{\sum\limits_{j \in {\mathcal {U}_R}} \!\!\!{( {H_{j}}(t) - {Q_{j}}(t)) \mathbb{E}[ {R_{j}^\pi}(t)] } } \\
~~~~~~~~~- \!\!\!\!\sum\limits_{m \in {\mathcal {U}_H}} {{Q_{m}}(t)\mathbb{E}[\mu_{m}^\pi(t)\tau] }- {\sum\limits_{j \in {\mathcal {U}_R}} {{Q_{j}}(t)\mathbb{E}[{\mu _{j}^\pi}(t)\tau]} }\\
~~~~~~~~~- {{Z}(t)(\mathbb{E}[\mu_{\rm sum}^\pi(t)] - W\eta_{\rm EE}^{\rm req}\mathbb{E}[{p_{\rm sum}^\pi}(t)] } ), \\
\end{array}
\label{randomdrift}
\end{equation}

By plugging (\ref{randompolicy1})-(\ref{randompolicy2}) into the
R.H.S. of (\ref{randomdrift}), we have
\begin{equation}
\mathbb{E}[ L(\chi (t + 1)) - L(\chi (t))] - V\mathbb{E}
[U({\bm{\gamma}}(t))] \le C - V{U^*}.
\end{equation}

Then by summing the above over slot $t \in \{0,1,...T-1\}$ and
dividing the result by $T$, we have
\begin{equation}
\frac{{\mathbb{E}[ L(\chi (t + 1))] - \mathbb{E}[ L(\chi (0))]
}}{T} - \frac{1}{T}\sum\limits_{t = 0}^{T - 1} {\mathbb{E}[
U({\bm{\gamma}}(t))}] \le C - V{U^*}.
\end{equation}

Considering the fact that $L(\chi (t + 1)) \ge 0$ and $L(\chi (0)) =
0$, we then have
\begin{equation}
\mathop {\lim }\limits_{T \to \infty } \frac{1}{T}\sum\limits_{t =
0}^{T - 1} {\mathbb{E}[U({\bm{\gamma}}(t))] } \ge {U^*} - C/V.
\end{equation}

Furthermore, since the utility function is a non-decreasing concave
function, according to Jensen's inequality, we finally have
\begin{equation}
U(\bar{\bf{r}}) \ge U(\bar{\bm{\gamma}}) \ge
\frac{1}{T}\sum\limits_{t = 0}^{T - 1} {\mathbb{E}[
U({\bm{\gamma}}(t)) ]} \ge {U^*} - C/V.
\end{equation}

\end{document}